\documentclass{article}
\usepackage{amsfonts}
\usepackage{amsmath}
\usepackage{subcaption}
\usepackage{graphicx}
\DeclareMathOperator*{\esssup}{ess\,sup}
\DeclareMathOperator*{\essinf}{ess\,inf}
\newtheorem{theorem}{Theorem}

\newtheorem{definition}[theorem]{Definition}

\newtheorem{lemma}[theorem]{Lemma}

\newtheorem{proposition}[theorem]{Proposition}
\newtheorem{remark}[theorem]{Remark}

\newtheorem{assumption}{Assumption}
\newenvironment{proof}[1][Proof]{\noindent \textbf{#1.} }{\  \rule{0.5em}{0.5em}}
\begin{document}

\title{An ergodic BSDE approach to forward entropic risk measures:
representation and large-maturity behavior\thanks{%
We thank the Editor, the Associate Editor, and two referees for
their comments and suggestions on the paper. This work was presented
at the SIAM Conference on Financial Mathematics and Engineering,
Austin; the 9th World Congress of the Bachelier Finance Society,
New York; the 9th European Summer School in Financial Mathematics,
Pushkin; and the 5th Berlin Workshop on Mathematical Finance for Young Researchers, Belrin.
The authors would like to thank the audience for
helpful comments and suggestions.}}
\author{W. F. Chong\thanks{%
Department of Statistics and Actuarial Science, The University of
Hong Kong, Pokfulam Road, Hong Kong; email:
\texttt{alfredcwf@hku.hk}. Department of
Mathematics, King's College London, London WC2R 2LS, U.K.; email: \texttt{%
wing\_fung.chong@kcl.ac.uk} } \and Y. Hu\thanks{%
IRMAR, Universit{\'e} Rennes 1, Campus de Beaulieu, 35042 Rennes
Cedex,
France; email: \texttt{ying.hu@univ-rennes1.fr} } \and G. Liang\thanks{%
Department of Mathematics, King's College London, London WC2R 2LS,
U.K.;
email: \texttt{gechun.liang@kcl.ac.uk} } \and T. Zariphopoulou\thanks{%
Department of Mathematics and IROM, The University of Texas at
Austin, U.S.A. and the Oxford-Man Institute, University of Oxford,
U.K.; email: \texttt{zariphop@math.utexas.edu}}}
\date{First version:\ January 2016; This version: \today}
\maketitle

\begin{abstract}
Using elements from the theory of ergodic backward stochastic
differential equations (BSDE), we study the behavior of forward
entropic risk measures. We provide their general representation
results (via both BSDE and convex duality) and examine their
behavior for risk positions of long maturities. We show that forward
entropic risk measures converge to some constant exponentially fast.
We also compare them with their classical counterparts and derive a
parity result.
\end{abstract}

\section{Introduction}

Risk measures constitute one of the most active areas of research in
financial mathematics, for they provide a general axiomatic
framework to assess risks. Their universality and wide
applicability, together with the wealth of related interesting
mathematical questions, have led to
considerable theoretical and applied developments (see, among others, \cite%
{Artzner_1999, Follmer_2010,Frittelli_2002} with more references
therein, and \cite{Detlefsen_2005, Kloppel_2007, Riedel_2004} for
dynamic convex risk measures).

A number of popular risk measures are defined in relation to
investment opportunities in a given financial market like, for
example, VaR, CVaR, indifference prices, etc. Such measures are
however tied to both a horizon and a market model, and these choices
are made at initial time with limited, if any, flexibility to be
revised.

As a result, issues related to how the risk of upcoming positions of
arbitrary maturities can be assessed, model revision can be
implemented, time-consistency can be preserved, etc. arise. Some of
these questions were addressed by one of the authors and Zitkovic in
\cite{ZZ}, where an axiomatic construction of the so-called
\textquotedblleft maturity-independent" risk measures was proposed.

Herein we analyze an important subclass of maturity-independent risk
measures, the \textit{forward entropic} ones. They are constructed
via the forward exponential performance criteria (see Definition \ref%
{def_forward_risk_measure}) and yield the risk assessment of a
position by comparing the optimal investment, under these criteria,
with and without the risk position. Like the forward performance
processes, via which they are built, the forward entropic measures
are defined for \textit{all} times.

In this paper, we focus on a stochastic factor model in an
incomplete market, which consists of multi-assets, and their pricing
dynamics depend on correlated stochastic factors (see (\ref{stock})
and (\ref{factor})). Stochastic factors are frequently used to model
the dynamics of assets (see, for example, the review paper
\cite{Zar}). In the forward performance setting, the use of
stochastic factors is discussed in \cite{Nadtochiy-Tehranchi}, where
the multi-stock/multi-factor complete market case is solved. The
incomplete market case with a single stock/single factor was
examined in \cite{Nadtochiy-Z} and, more recently, in
\cite{Shkolnikov-Sircar-Z} for a model with slow and fast stochastic
factors.

Our contribution is threefold. Firstly, we provide a general
backward stochastic differential equation (BSDE) representation
result for forward entropic risk measures. We do so by building on a
recent work of two of the authors (\cite{LZ}), who developed a new
approach for the construction of homothetic (exponential, power and
logarithmic) forward performance processes using elements from the
theory of ergodic BSDE. This method bypasses a number of technical
difficulties associated with solving an underlying ill-posed
stochastic partial differential equation (SPDE) that the forward
process is expected to satisfy (see \cite{ElKaroui} and \cite{MZ3}
for the discussion on the corresponding SPDE).

For the exponential forward family we consider herein, the approach in \cite%
{LZ} yields the forward performance criterion in a factor
form (see (\ref{ExponentialForwardUtility}%
)), for which the representation is unique. Having this
representation result, we show that the risk measure satisfies a
BSDE\ whose driver, however, depends on the solution of the
aforementioned ergodic BSDE (see Theorem \ref{theorem}), and such a
solution can be linked to the \emph{volatility} of the forward
process (see Remark \ref{remark_1}). The fact that the risk measure
is computed via a BSDE is not surprising since we are dealing with
the valuation of a risk position, which is by nature a
\textquotedblleft backward" problem. However, the new element is the
dependence of the driver of this BSDE\ on the solution of another,
actually ergodic, BSDE.

We then show two applications of the BSDE representation result. By
using the convex property of the BSDE driver, we derive a convex
dual representation of forward entropic risk measures (see Theorem
\ref{theorem_dual}). More specifically, we show that the forward
entropic risk measure is the minimal expected value of the risk
position (subject to a penalty term) over all equivalent probability
measures. Such a penalty term, depending solely on the stochastic
factor process and the volatility of the forward performance
process, is the convex dual of the BSDE driver. As a consequence, we
obtain several properties of forward entropic risk measures, namely,
{anti-positivity}, {convexity} and {cash-translativity}. In a single
stock/single stochastic factor case, we further show that the
equivalent probability measures are actually equivalent martingale
measures (see section 6).

We then study the asymptotic behavior of forward entropic risk
measures when the maturity of the risk position is very long. For
risk positions that are deterministic functions of the stochastic
factor process, we show that their risk measure converges to a
constant, which is independent of the initial state of the
stochastic factors, and, furthermore, we prove that the convergence
is exponentially fast. As a consequence, we derive an explicit
exponential bound of the investor's hedging strategies. In
particular, when the maturity goes to infinity, we show that the
investor will not do any trading to hedge the underlying risks in
any finite time (see Theorem \ref{thm:long_term_price}).

Our third contribution is a derivation of a parity result between
the forward and the classical entropic risk measures (see
Proposition \ref{proposition14}). We show that the forward measure
can be constructed as the difference of two classical entropic
measures applied to the risk position, and to a normalizing factor
related to the solution of the ergodic BSDE\ for the forward
criterion.

We conclude with an example cast in the single stock/single
stochastic factor case. Using the ergodic BSDE\ approach, we derive
a closed form representation of the forward risk measure (see
(\ref{closed-form})) and its convex dual representation (see
(\ref{convex_formula })). We
also derive a representation of the classical entropic risk measure (see (%
\ref{closed-form-backward})) and, in turn, compute numerically the
long-term limits of the two measures for specific risk positions.

The paper is organized as follows. In section 2, we introduce the
stochastic factor model and provide background results on
exponential forward performance processes. In section 3, we provide
general representation results of forward entropic risk measures
using both BSDE and convex duality. We also study their behavior for
risk positions of long maturities. In section 4, we present the
proofs of the main results and, in section 5, we derive the parity
result. We conclude in section 6 with an example.


\section{The stochastic factor market model}

Let $W$ be a $d$-dimensional Brownian motion on a probability space
$(\Omega ,\mathcal{F},\mathbb{P})$. Denote by $\mathbb{F}=\{
\mathcal{F}_{t}\}_{t\geq 0}$ the augmented filtration generated by
$W$. We consider a market of a risk-free bond offering zero interest
rate and $n$ risky stocks, with $n\leq d$. The stock price processes
$S_{t}^{i}$, $t\geq 0$, solve, for $i=1,\ldots ,n$,
\begin{equation}
\frac{dS_{t}^{i}}{S_{t}^{i}}=b^{i}(V_{t})dt+\sigma
^{i}(V_{t})dW_{t}\text{.} \label{stock}
\end{equation}%
The $d$-dimensional stochastic process $V$ models the stochastic
factors affecting the coefficients of the stock prices, and solves
\begin{equation}
dV_{t}=\eta (V_{t})dt+\kappa dW_{t}\text{.}  \label{factor}
\end{equation}

We introduce the following model assumptions. Throughout, we will be
using the superscript $A^{tr}$ to denote the transpose of matrix
$A$.

\begin{assumption}
\label{assumption1} i) The drift and volatility coefficients,
$b^{i}(v)\in \mathbb{R}$ and $\sigma ^{i}(v)\in \mathbb{R}^{1\times
d}$ are uniformly bounded in $v\in \mathbb{R}^d$.

ii)\ The volatility matrix, defined as $\sigma (v)=(\sigma
^{1}(v),\ldots ,\sigma ^{n}(v))^{tr}$, has full row rank $n.$

iii)\ The market price of risk, defined as
\begin{equation}
\theta (v):=\sigma (v)^{tr}[\sigma (v)\sigma (v)^{tr}]^{-1}b(v),
\label{market-price-risk}
\end{equation}%
$v\in \mathbb{R}^{d}$, is uniformly bounded and Lipschitz
continuous.
\end{assumption}

Note that $\theta (v)$ solves $\sigma (v)\theta (v)=b(v)$, $v\in \mathbb{R}%
^{d},$ which admits a solution because of (ii) above.

\begin{assumption}
\label{assumption2} i) The drift coefficients $\eta (v)\in
\mathbb{R}^{d}$ satisfy a dissipative condition, namely, there
exists a large enough constant $C_{\eta }>0$ such that, for
$v,\bar{v}\in \mathbb{R}^{d}$,
\begin{equation*}
(\eta (v)-\eta (\bar{v}))^{tr}(v-\bar{v})\leq -C_{\eta }|v-\bar{v}|^{2}\text{%
.}
\end{equation*}

ii) The volatility matrix $\kappa \in \mathbb{R}^{d\times d}$ is
positive definite and normalized to $|\kappa |=1$.
\end{assumption}

The \textquotedblleft large enough" property will be quantified in
the
sequel when it is assumed that $C_{\eta }>C_{v}>0,$ where the constant $%
C_{v} $ appears in the properties of the driver of an upcoming
ergodic BSDE (see inequality (\ref{driver0})).

Under Assumption 2, the stochastic factor process $V$ admits a
unique invariant measure and it is thus ergodic. Moreover, any two
paths will converge to each other exponentially fast.

In this market environment, an investor trades dynamically among the
risk-free bond and the risky assets. Let $\tilde{\pi}=(\tilde{\pi}%
^{1},\cdots ,\tilde{\pi}^{n})^{tr}$ denote the (discounted by the
bond) amounts of her wealth in the stocks, which are taken to be
self-financing. Then, the (discounted by the bond) wealth process
satisfies
\begin{equation*}
dX_{t}^{\tilde{\pi}}=\sum_{i=1}^{n}\frac{\tilde{\pi}_{t}^{i}}{S_{t}^{i}}%
dS_{t}^{i}=\tilde{\pi}_{t}^{tr}\sigma (V_{t})\left( \theta
(V_{t})dt+dW_{t}\right) ,
\end{equation*}%
with $X_{0}=x\in \mathbb{R}.$ As in \cite{LZ}, we work with the
investment strategies rescaled by the volatility,
\begin{equation*}
\pi _{t}^{tr}:=\tilde{\pi}_{t}^{tr}\sigma (V_{t}).
\end{equation*}%
Then, the wealth process satisfies
\begin{equation}
dX_{t}^{\pi }=\pi _{t}^{tr}\left( \theta (V_{t})dt+dW_{t}\right) .
\label{wealth}
\end{equation}%
For any $t\geq 0$, we denote by $\mathcal{A}_{[0,t]}$ the set of the
admissible investment strategies in $[0,t],$ defined as
\begin{equation*}
\mathcal{A}_{[0,t]}=\{ \pi \in \mathcal{L}_{BMO}^{2}[0,t]:\  \pi
_{s}\in \Pi \  \text{for}\ s\in \lbrack 0,t]\} \text{,}
\end{equation*}%
where $\Pi$ is a closed and convex subset in $\mathbb{R}^{d}$ also
including the origin $0$, and
\begin{eqnarray*}
&&\mathcal{L}_{BMO}^{2}[0,t]=\left \{ \pi _{s},\text{ }s\in \left[
0,t\right] :\pi \  \text{is}\  \mathbb{F}\text{-progressively\
measurable and}\right.\\
&&\esssup_{\tau }E_{\mathbb{P}}\left[\left. \int_{\tau }^{t}|\pi
_{s}|^{2}ds\right\vert%
\mathcal{F}_{\tau }\right] <\infty, \  \text{for any}\  \mathbb{F}\text{%
-stopping time}\  \tau \in \lbrack 0,t]\} \text{.}
\end{eqnarray*}%
The set of admissible investment strategies for \textit{all} $t\geq
0$ is, in turn, defined as $\mathcal{A}=\cup _{t\geq
0}\mathcal{A}_{[0,t]}$.\\

The investor has an exponential forward performance criterion for
her admissible investment strategies. For the reader's convenience,
we start with some background results on the forward performance
criterion. We
first recall its definition (see \cite%
{MZ0}-\cite{MZ2}) and, in turn, focus on the exponential class. We
then recall its ergodic BSDE representation, established in
\cite{LZ}.

\begin{definition}
\label{def:forward_performance} \label{def} Let $\mathbb{D=R}\times
\left[ 0,\infty \right)$. A process $U\left( x,t\right) $, $\left(
x,t\right) \in \mathbb{D}$, is a forward performance process if:

i) for each $x\in \mathbb{R}$, $U\left( x,t\right) $ is $\mathbb{F}$%
-progressively measurable,

ii) for each $t\geq 0$, the mapping $x\mapsto U(x,t)$ is strictly
increasing and strictly concave,

iii) for all $\pi \in \mathcal{A}$ and $0\leq t\leq s$,
\begin{equation*}
U( X_{t}^{\pi },t) \geq E_{\mathbb{P}}\left[ U(X_{s}^{\pi },s)|%
\mathcal{F}_{t}\right] ,
\end{equation*}%
and there exists an optimal $\pi ^{\ast }\in \mathcal{A}$ such that,
\begin{equation*}
U( X_{t}^{\pi ^{\ast }},t) =E_{\mathbb{P}}\left[ U(X_{s}^{\pi ^{\ast
}},s)|\mathcal{F}_{t}\right] ,
\end{equation*}%
with $X^{\pi },X^{\pi ^{\ast }}$solving (\ref{wealth}).
\end{definition}

Throughout the paper, we work with \emph{Markovian exponential
forward criteria} that are appropriate functions of the stochastic
factor process $V,$ namely,
\begin{equation}
U\left( x,t\right)=-e^{-\gamma x+f(V_{t},t)}\text{,}
\label{exponential-forward-general}
\end{equation}%
for $\left( x,t\right) \in \mathbb{D},$ with $\gamma >0$, and the function $%
f:\mathbb{R}^{d}\times \left[ 0,\infty \right) \rightarrow
\mathbb{R}$ to be determined. The main idea is to use the Markovian
solution of an ergodic BSDE to construct the function $f$. We refer
to \cite{Pham, HU1, HU2, Hu11, LZ} for recent developments of
ergodic BSDE. In {Proposition 4.1} of \cite{LZ}, the following
result was proved.

\begin{proposition}
\label{prop:ergoic_exist} Suppose that Assumptions \ref{assumption1} and \ref%
{assumption2} hold. Then, the ergodic BSDE
\begin{equation}
dY_{t}=(-F(V_{t},Z_{t})+\lambda )dt+Z_{t}^{tr}dW_{t},
\label{EQBSDE2}
\end{equation}%
where the driver $F:\mathbb{R}^{d}\times \mathbb{R}^{d}\rightarrow
\mathbb{R}$ is defined as
\begin{equation}
F(v,z):=\frac{1}{2}\gamma ^{2}dist^{2}\left \{ \Pi ,\frac{z+\theta (v)}{%
\gamma }\right \} -\frac{1}{2}|z+\theta
(v)|^{2}+\frac{1}{2}|z|^{2}\text{,} \label{driver}
\end{equation}%
with $\theta \left( \cdot \right) $ as in (\ref{market-price-risk}),
admits a unique Markovian solution $(Y_{t},Z_{t},\lambda ),$\ $t\geq
0$. Specifically, there
exist a unique $\lambda \in \mathbb{R},$ and functions $y:\mathbb{R}%
^{d}\rightarrow \mathbb{R}$ and $z:\mathbb{R}^{d}\rightarrow
\mathbb{R}^{d}$ such that
\begin{equation}
\left( Y_{t},Z_{t}\right) =\left( y\left( V_{t}\right) ,z\left(
V_{t}\right) \right) .  \label{Y-Z}
\end{equation}%
The function $y(\cdot )$ has at most linear
growth with $y(0)=0$, and $z(\cdot )$ is bounded with $|z(\cdot )|\leq \frac{C_{v}}{%
C_{\eta }-C_{v}}$, where $C_{\eta }$ and $C_{v}$ are as in Assumption \ref%
{assumption2} and equality (\ref{driver0}), respectively. Moreover,
\begin{equation}
|\nabla y(v)|\leq \frac{C_{v}}{C_{\eta }-C_{v}}\text{.}
\label{gradient_for_y}
\end{equation}
\end{proposition}

The solution $y(\cdot)$ is unique up to a constant (i.e.
$y(\cdot)+C$ for any constant $C$ is also a solution), and to avoid
this ambiguity, we fix without loss of generality its value at $0$
as $y(0)=0$.

We stress that while there exists a unique solution in factor form,
the ergodic BSDE (\ref{EQBSDE2}) admits \textit{multiple
non-Markovian} solutions, which are however not considered herein.\\

The next result relates the solution of the ergodic BSDE
(\ref{EQBSDE2}) to the exponential forward performance process (\ref%
{exponential-forward-general}) and its associated optimal policy.
For its proof, see Theorem 4.2 of \cite{LZ}.

\begin{proposition}
Suppose that Assumptions \ref{assumption1} and \ref{assumption2}
hold, and let $(Y,Z,\lambda )$ be the unique Markovian solution to
(\ref{EQBSDE2}). Then, the process $U(x,t),$ $\left( x,t\right) \in
\mathbb{D},$ given by
\begin{equation}
U(x,t)=-e^{-\gamma x+}{}^{Y_{t}-\lambda t}\text{,}
\label{ExponentialForwardUtility}
\end{equation}%
is an exponential forward performance process, and the associated
optimal strategy
\begin{equation}
\pi _{t}^{\ast }=\text{\textit{Proj}}_{\Pi }\left( \frac{\theta (V_{t})}{%
\gamma }+\frac{Z_{t}}{\gamma }\right) \text{.}
\label{eq:forward_strategy}
\end{equation}
\end{proposition}

\begin{remark}\label{remark_1} In \cite{LZ}, the solution pair $(Y,Z)$ is constructed
by a ``vanishing discount rate" argument, i.e. $(Y,Z)$ are the
limiting processes of the solution to an infinite horizon BSDE with
a discount factor $\rho$ as $\rho\downarrow 0$.

The solution $Z$ can be regarded as the volatility of the forward
performance process. To see this, an application of It\^o's formula
to $e^{-\gamma x+Y_t-\lambda t}$ yields that
\begin{equation*}
\frac{dU(x,t)}{U(x,t)}=\left(-F(V_t,Z_t)+\frac12|Z_t|^2\right)dt+(Z_t)^{tr}dW_t.
\end{equation*}
Furthermore, the solution $\lambda$ can be interpreted as the long
term growth rate of a classical utility maximization problem (see
Proposition 3.3 in \cite{LZ}).
\end{remark}

\section{Forward entropic risk measures and ergodic BSDE}

We now study the behavior of forward entropic risk measures in
relation to the exponential forward performance process
(\ref{ExponentialForwardUtility}). Forward entropic risk measures
were first introduced in \cite{ZZ} to assess risk positions with
arbitrary maturities.

We introduce the space of candidate risk positions,%
\begin{equation}
\mathcal{L=\cup }_{T\geq 0}\mathcal{L}^{\infty }\left( \mathcal{F}%
_{T}\right), \label{L-space}
\end{equation}%
where $\mathcal{L}^{\infty }(\mathcal{F}_{T})$ is the space of
uniformly bounded $\mathcal{F}_{T}$-measurable random variables. To
simplify the presentation, we assume without loss of generality that
the generic risk position is introduced at time $t=0,$ and also
remind the reader that $\mathbb{D}=\mathbb{R\times }\left[ 0,\infty
\right) .$

\begin{definition}
\label{def_forward_risk_measure} i) Consider the forward exponential
performance process $U\left( x,t\right) =-e^{-\gamma
x+}{}^{Y_{t}-\lambda
t}, $ $\left( x,t\right) \in \mathbb{D}$ (cf. (\ref%
{ExponentialForwardUtility})).

Let
$T>0$ be arbitrary and consider a risk position $\xi _{T}\in
\mathcal{L}^{\infty }\left( \mathcal{F}_{T}\right)$. Its
($T$-normalized) forward entropic risk measure, denoted by $\rho
_{t}(\xi _{T};T)\in\mathcal{F}_t$, is defined by
\begin{equation}\label{def_1}
{u}^{\xi_T}(x+\rho_t(\xi_T;T),t)=U(x,t),
\end{equation}
for all $\left( x,t\right) \in \mathbb{R}\times[0,T],$ where
\begin{equation}
{u}^{\xi_T}(x,t):=\esssup_{\pi \in
\mathcal{A}_{[t,T]}}E_{\mathbb{P}}\left[ \left. U(x+\int_{t}^{T}\pi
_{u}^{tr}\left( \theta (V_{u})du+dW_{u}\right) +\xi _{T},T)\right
\vert \mathcal{F}_{t}\right],
\end{equation}%
provided the above essential supremum exists.

ii) Let the risk position $\xi \in \mathcal{L}$. Define $T_{\xi
}=\inf \{T\geq 0:\xi \in \mathcal{F}_{T}\}.$ Then, the forward
entropic risk measure of $\xi $ is defined, for $t\in \left[
0,T_{\xi }\right] ,$ as
\begin{equation}
\rho _{t}(\xi ):=\rho _{t}(\xi;T_{\xi }).
\label{maturity-independent}
\end{equation}
Hence, for $\xi_T\in \mathcal{L}^{\infty}(\mathcal{F}_T)$, we have $%
\rho_t(\xi_T)=\rho_t(\xi_T;T)$.
\end{definition}

We stress that the performance criterion entering in Definition
\ref{def_forward_risk_measure} is defined for \textit{all }$T>0,$
and thus one can assess risk positions with \emph{arbitrary
maturities}. This is not, however, the case in the classical
framework. In the latter, the performance criterion is built on a
(static) utility function defined for a \emph{single} time $T$, and,
as a result, one can only assess risk positions at that time (see
Definition \ref{def_classical_risk_measure} in section
\ref{section:classical}).

\subsection{BSDE representation of forward entropic risk measures}

We are now ready to provide one of the main results herein, which is
the representation of forward entropic risk measures using the
solutions of associated BSDE and ergodic BSDE. The main idea is to
express the forward entropic risk measure process as the solution of
a traditional BSDE whose driver, however, {depends} on the
\emph{volatility} of the forward performance process, i.e.
the solution $Z_t$ of the ergodic BSDE\ (\ref%
{EQBSDE2}). This dependence follows from the fact that equation
(\ref{EQBSDE2})\ was used to construct the exponential forward
performance process (\ref{ExponentialForwardUtility}) that appears
in (\ref{def_1}) in Definition \ref{def_forward_risk_measure}.

\begin{theorem}
\label{theorem} Consider a risk position $\xi _{T}\in \mathcal{L}^{\infty}(%
\mathcal{F}_{T}),$ with its maturity $T>0$ being arbitrary. Suppose
that Assumptions \ref{assumption1} and \ref{assumption2} hold.
Consider, for $t\in \left[ 0,T\right] ,$ the BSDE
\begin{equation}
Y_{t}^{\xi_T}=\xi_{T}+\int_{t}^{T}G(V_{u},Z_{u},Z_{u}^{\xi_T
})du-\int_{t}^{T}(Z_{u}^{\xi_T})^{tr}dW_{u}\text{,}
\label{BSDERepresentation}
\end{equation}%
where the driver $G:\mathbb{R}^{d}\times \mathbb{R}^{d}\times \mathbb{R}%
^{d}\rightarrow \mathbb{R}$ is defined as
\begin{equation}
G(v,{z},\bar{z}):=\frac{1}{\gamma }\left( F(v,{z}+\gamma \bar{z})-F(v,{z}%
)\right) \text{,}  \label{driver_2}
\end{equation}%
with $F(\cdot ,\cdot )$ given by (\ref{driver}). Then the following
assertions hold:

i) The BSDE (\ref{BSDERepresentation}) has a unique solution $\left(
Y_{t}^{\xi_T},Z_{t}^{\xi_T}\right) ,$ $t\in \left[ 0,T\right] ,$ with $%
Y^{\xi_T}$ being uniformly bounded and $Z^{\xi_T}\in \mathcal{L}%
_{BMO}^{2}[0,T]$.

ii) The forward entropic risk measure of $\xi _{T}$ is given, for
$t\in \left[ 0,T\right] ,$ by
\begin{equation}
\rho _{t}(\xi _{T})=Y_{t}^{-\xi_T}\text{.}  \label{representation}
\end{equation}
\end{theorem}

From the above representation, we readily obtain the
time-consistency property: for any $0\leq t\leq s\leq T<\infty$,
\begin{equation*}
\rho _{t}(\xi _{T})=Y_{t}^{-\xi_T}=Y_{t}^{Y_{s}^{-\xi_T
}}=Y_{t}^{\rho _{s}(\xi _{T})}=\rho _{t}(-\rho _{s}(\xi _{T})).
\end{equation*}
Therefore, the forward entropic risk measure is well defined over
all time horizons.

The fact that the forward entropic risk measure\textit{\ }is
obtained via a BSDE\ (cf. (\ref{BSDERepresentation})) should
\textit{not} suggest that it does not differ from its classical
counterpart, which is also given as a solution of a BSDE (see
(\ref{BSDE_11}) herein).

Firstly, it is natural to expect that the risk measure will be
represented by the solution of a BSDE,\ since the pricing condition
(\ref{def_1}) is, by nature, set \textquotedblleft backwards" in
time. However, the classical entropic risk measure is defined only
on $\left[ 0,T\right]$ for a \emph{single} maturity $T$, because the
associated traditional exponential value function is only defined on
$t\in\left[ 0,T\right] $. In contrast, the forward entropic risk
measure\textit{\ }is defined for \textit{all} maturities $T\geq 0,$
for the
associated forward process $U\left( x,t\right) $ (cf. (\ref%
{ExponentialForwardUtility})) is defined for all $t\geq 0$.

Secondly, the BSDE\ (\ref{BSDERepresentation})\ for the forward risk
measure differs from the one for the classical entropic measure,
because its driver depends on the {volatility process} $Z_{t},$
which solves the ergodic BSDE that yields the exponential forward
criterion.

\begin{remark}\label{remark_strategy}
In addition to forward entropic measures, one can define the
\textquotedblleft \textit{hedging strategies}" associated with the
risk position $\xi _{T}$. As in the classical case, they are defined
as the difference of the optimal strategies for ${u}^{\xi_T}(x,t)$
and $U(x,t)$ appearing in Definition \ref{def_forward_risk_measure}.
We deduce that the related hedging strategy, denoted by $\alpha
_{t,T},t\in \lbrack 0,T],$ is given by (cf. (\ref{aux-optimal}) and
(\ref{eq:forward_strategy}))
\begin{equation}
\alpha _{t,T}={\pi }_{t}^{\ast ,\xi_T}-{\pi }_{t}^{\ast }
\label{hedging-strategy}
\end{equation}%
\begin{equation*}
=\text{\textit{Proj}}_{\Pi }\left( Z_{t}^{-\xi_T}+\frac{Z_{t}+\theta (V_{t})%
}{\gamma }\right) -\text{\textit{Proj}}_{\Pi }\left(
\frac{Z_{t}+\theta (V_{t})}{\gamma }\right) .
\end{equation*}

Observe that the first term naturally depends on the maturity of the
risk position, while the second is independent of it and defined for
all times. This is not the case in the classical setting, where both
terms depend on the investment horizon.
\end{remark}

\subsection{Convex dual representation of forward entropic risk measures}

Our second main result is a dual representation of the forward
entropic risk measure $\rho_t(\xi_T)$ via the BSDE
(\ref{BSDERepresentation}). Firstly, we observe that since $\Pi$ is
convex, and a distance function to a convex set is also convex, it
follows that the driver $G(v,z,\bar{z})$, defined in
(\ref{driver_2}), is convex in $\bar{z}$, for any
$(v,z)\in\mathbb{R}^d\times\mathbb{R}^d$. We can therefore introduce
the convex dual of $G(v,z,\bar{z})$,
\begin{equation}\label{dual}
{G}^*(v,z,q):=\sup_{\bar{z}\in\mathbb{R}^d}\left(\bar{z}^{tr}q-G(v,z,\bar{z})\right),
\end{equation}
for $q\in\mathbb{R}^d$. Note that $G^*$ is valued in
$\mathbb{R}\cup\{\infty\}$.

Then, the Fenchel-Moreau theorem yields that
\begin{equation}\label{FM_theorem}
G(v,z,\bar{z})=\sup_{q\in\mathbb{R}^d}(\bar{z}^{tr}q-G^{*}(v,z,q)),
\end{equation}
{for} $\bar{z}\in\mathbb{R}^d$. Moreover, $q^*\in\partial
G_{\bar{z}}(v,z,\bar{z})$, which is the subdifferential of
$\bar{z}\mapsto G(v,z,\bar{z})$ at $\bar{z}\in\mathbb{R}^d$,
achieves the supremum in (\ref{FM_theorem}),
\begin{equation}\label{FM_theorem_2}
G(v,z,\bar{z})=\bar{z}^{tr}q^*-G^{*}(v,z,q^*).
\end{equation}

With a slight abuse of notation, for any
$q\in\mathcal{L}^2_{BMO}[0,T]$, its stochastic exponential
$\mathcal{E}(\int_0^{\cdot}q^{tr}_sdW_s)$ is a uniformly integrable
martingale, since $\int_0^{\cdot}q^{tr}_sdW_s$ is a
\emph{BMO}-martingale. We then define on $\mathcal{F}_T$ a
probability measure $\mathbb{Q}^{q}$ by
$\frac{d\mathbb{Q}^{q}}{d\mathbb{P}}=\mathcal{E}(\int_0^{\cdot}q^{tr}_sdW_s)_T$,
and introduce the admissible set
\begin{equation*}
\mathcal{A}^*_{[0,T]}=\left\{ q \in \mathcal{L}_{BMO}^{2}[0,T]:
E_{\mathbb{Q}^q}\left[\int_0^T|G^*(V_s,Z_s,q_s)|ds\right]<\infty\right\}
\text{,}
\end{equation*}%
where $V$ is the stochastic factor process given in (\ref{factor}),
and $Z$ appears in the unique Markovian solution of the BSDE
(\ref{EQBSDE2}).

\begin{theorem}\label{theorem_dual}
Suppose that Assumptions \ref{assumption1} and \ref{assumption2}
hold. Let $\xi _{T}\in \mathcal{L}^{\infty}(%
\mathcal{F}_{T})$ be a risk position with its maturity $T>0$ being
arbitrary. Then, the following assertions hold:

i) The forward entropic risk measure $\rho_t(\xi_T)$ admits the
following convex dual representation
\begin{equation}\label{dual_formula}
\rho_t(\xi_T)=-\essinf_{q\in\mathcal{A}^*_{[t,T]}}E_{\mathbb{Q}^q}
\left[\left.\xi_T+\int_t^{T}{G}^*(V_s,Z_s,q_s)ds\right\vert\mathcal{F}_t\right],
\end{equation}
where
${G}^*:\mathbb{R}^{d}\times\mathbb{R}^{d}\times\mathbb{R}^{d}\rightarrow\mathbb{R}\cup\{\infty\}$
is the convex dual of $G$ in (\ref{dual}).

ii) There exists an optimal density process
$q^*\in\mathcal{A}^*_{[t,T]}$ such that
\begin{equation}\label{dual_formula_2}
\rho_t(\xi_T)=-E_{\mathbb{Q}^{q^*}}
\left[\left.\xi_T+\int_t^{T}{G}^*(V_s,Z_s,q_s^*)ds\right\vert\mathcal{F}_t\right].
\end{equation}
\end{theorem}

We note that the penalty term $G^*$ depends solely on the stochastic
factor process $V_s$ and the volatility process $Z_s$ of the forward
performance process. In section 6, we provide a concrete example to
further illustrate the structure of the penalty term $G^*$.

\begin{remark} From the convex dual representation
(\ref{dual_formula}), we easily deduce the following properties of
the forward entropic risk measure $\rho_t(\xi_T)$.

i) Anti-positivity: $\rho_t\left( \xi_T \right) \leq 0$ for any
$\xi_T\in\mathcal{L}^{\infty}(\mathcal{F}_T)$ with $\xi_T \geq 0,$
which follows from the nonnegativity of $G^*$ (cf.
(\ref{bound_G_hat})).

ii) Convexity: $\rho_t\left( \alpha \xi_T +(1-\alpha \right)
\bar{\xi}_T)\leq \alpha \rho_t ( \xi_T )+(1-\alpha) \rho_t
(\bar{\xi}_T)$, for any $\alpha\in[0,1]$ and
$\xi_T,\bar{\xi}_T\in\mathcal{L}^{\infty}(\mathcal{F}_T)$.

iii) Cash-translativity: $\rho_t \left( \xi_T -m\right) =\rho_t
\left( \xi_T \right) +m$, for any $m\in \mathbb{R}$ and
$\xi_T\in\mathcal{L}^{\infty}(\mathcal{F}_T)$.
\end{remark}

\subsection{Long-maturity behavior of forward entropic risk measures}

\label{section_4}

We study the behavior of forward entropic risk measures when the
maturity of the risk position is long. We focus on European-type
positions written only on the stochastic factor process $V$.
Specifically, with $T>0$ being arbitrary, we consider risk positions
represented as
\begin{equation}
\xi _{T}=-g\left( V_{T}\right) ,  \label{payoff-european}
\end{equation}%
with $g:\mathbb{R}^{d}\rightarrow \mathbb{R}$ being a uniformly
bounded and Lipschitz continuous function with Lipschitz constant
$C_{g}.$

From Theorem \ref{theorem}, we have the representation $\rho
_{t}\left( \xi _{T}\right) =Y_{t}^{-\xi_T}$ for $t\in[0,T]$.
Furthermore, by the Markovian assumption on the risk position
$\xi_T$, we have that
\begin{equation*}
(Y_t^{-\xi_T},Z_t^{-\xi_T})=(y^{T,g}(V_t,t),z^{T,g}(V_t,t)),
\end{equation*}
for some measurable functions
$(y^{T,g}(\cdot,\cdot),z^{T,g}(\cdot,\cdot))$.

We show that as $T\uparrow \infty ,$ the forward entropic risk
measure, from the one hand, has an exponential decay with respect to
its maturity, and from the other, tends to a constant, which is
independent of the initial state of the stochastic factor process
$V^v_0=v$. As a consequence, we derive an explicit exponential bound
of the investor's optimal investment strategy. In particular, we
show that the investor will not do any trading to hedge the
underlying risks in any finite time when the maturity goes to
infinity.

\begin{theorem}
\label{thm:long_term_price} Suppose that Assumptions \ref{assumption1} and %
\ref{assumption2} hold. Consider a risk position $\xi _{T}$ as in
(\ref{payoff-european}) with $T$ arbitrary. Then, the following
assertions hold:

i) There exists a constant $L^g\in \mathbb{R}$, independent of
$V_0^{v}=v$, such that
\begin{equation}
\lim_{T\uparrow \infty }\rho _{0}(\xi _{T})=\lim_{T\uparrow \infty
}y^{T,g}(v,0)=L^g \text{.}  \label{limit-Markovian}
\end{equation}%
%
%
%
%
Moreover, for any $T>0,$ \
\begin{equation}
|y^{T,g}(v,0)-L^g|\leq C(1+|v|^{2})e^{-\hat{C}_{\eta }T}\text{,}
\label{convergence_rate}
\end{equation}%
with the constant $\hat{C}_{\eta }$ given in Proposition \ref%
{prop:forward_estimate} (see Appendix A.1).

ii) The hedging strategy satisfies, for any $T>0$ and $s\in [0,T)$,
\begin{equation}\label{exponential_bound}
E_{\mathbb{P}}\left[ \int_{0}^{s}|\alpha _{t,T}|^{2}dt\right] \leq C
(1+|v|^4) e^{-2\hat{C}_{\eta}(T-s)}.
\end{equation}
Therefore, for any $s\in[0,T)$,
\begin{equation}
\lim_{T\uparrow \infty }E_{\mathbb{P}}\left[ \int_{0}^{s}|\alpha
_{t,T}|^{2}dt\right] =0.  \label{hedging_strategy}
\end{equation}
\end{theorem}

\begin{remark}\label{Remark} In order to study the long-maturity behavior of the
solution $Y^{-\xi_T}_0$ to the BSDE (\ref{BSDERepresentation}), it
is natural to relate (\ref{BSDERepresentation}) with an ergodic
BSDE, given below, and investigate the proximity of their solutions.

To this end, we may consider the ergodic BSDE
\begin{equation}
{P}_{t}={P}_{s}+\int_{t}^{s}\left( G(V_{u},Z_{u},%
{Q}_{u})-\hat{\lambda}\right) ds-\int_{t}^{s}({Q}%
_{u})^{tr}dW_{u}  \label{BSDERepresentation 22}
\end{equation}%
for $0\leq t\leq s<\infty $, {and examine the approximation of }$%
Y_{0}^{-\xi_T}$ by $P_{0}+\hat{\lambda} T$, for large $T.$

We stress that the driver of the ergodic BSDE
(\ref{BSDERepresentation 22}) depends on the solution $Z$ of the
ergodic BSDE (\ref{EQBSDE2}) of the forward performance process.
This causes various technical issues. The driver $G(v,z(v),\bar{z})$ of the ergodic BSDE (\ref%
{BSDERepresentation 22}) \textit{depends} on the function $z(v)$.
Although, due to the boundedness of the function $z(\cdot )$, the
driver $G$ satisfies the locally Lipschitz estimate (\ref{driver1})
in $\bar{z}$, it may not satisfy the locally Lipschitz estimate
(\ref{driver0}) in $v$, and hence the existence and uniqueness
result in \cite{LZ} might not apply. Moreover, it is not even clear
whether the ergodic BSDE (\ref{BSDERepresentation 22}) is well-posed
or not. For this, as we mention in the proof of Theorem
\ref{thm:long_term_price}, we work with the function
$y^{T,g}(\cdot,\cdot)$ directly.
\end{remark}

\section{Proofs of the main results}
\subsection{Proof of Theorem \ref{theorem}}

We first recall that in the proof of Proposition 4.1 in \cite{LZ},
two key inequalities were used, which follow from Assumption
\ref{assumption1} and the Lipschitz property of the distance
function. Specifically, it can be shown that there exist constants
$C_{v}>0$ and $C_{z}>0$ such that
\begin{equation}
|F(v_1,z)-F({v}_2,z)|\leq C_{v}(1+|z|)|v_1-{v}_2|\label{driver0}
\end{equation}%
and
\begin{equation}
|F(v,z_1)-F(v,{z}_2)|\leq C_{z}(1+|z_1|+|{z}_2|)|z_1-{z}_2|\text{,}
\label{driver1}
\end{equation}%
for any $v,v_1,{v}_2,z,z_1,{z}_2\in \mathbb{R}^{d}$.

\textbf{Proof of (i)}: First note that, for $t\in \left[ 0,T\right]
,$ $G(v,Z_{t},\bar{z})$ is locally Lipschitz continuous in
$\bar{z},$ since a.s.
\begin{equation*}
|G(v,Z_{t},\bar{z}_{1})-G(v,Z_{t},\bar{z}_{2})|\leq C_{z}(1+2|Z_{t}|+\gamma |%
\bar{z}_{1}|+\gamma |\bar{z}_{2}|)|\bar{z}_{1}-\bar{z}_{2}|\text{,}
\end{equation*}%
with $Z$ being uniformly bounded. Using, furthermore, that $\xi
_{T}\in \mathcal{L}^{\infty }\left( \mathcal{F}_{T}\right) $, the
assertion follows from Theorems 2.3 and 2.6 of \cite{Kobylanski} and
Theorem 7 of \cite{HU0}.

\textbf{Proof of (ii)}: Using (\ref{ExponentialForwardUtility}) and
that\textbf{\ }$\rho _{t}(\xi _{T})\in \mathcal{F}_{t},$ $t\in
\left[ 0,T\right] ,$ we have that
\begin{equation}
{u}^{\xi_T}(x+\rho_{t}(\xi_T),t)=e^{-\gamma \rho _{t}(\xi
_{T})}{u}^{\xi_T}(x,t)\text{.} \label{stochastic_control}
\end{equation}%
To obtain ${u}^{\xi_T}(x,t)$, we work as follows. Define, for $s\in
\lbrack t,T],$ the process
\begin{equation}
R_{s}^{\pi }:=-\exp\left(-\gamma \left( x+\int_{t}^{s}\pi _{u}^{tr}(
\theta (V_{u})du+dW_{u}) \right) +Y_{s}-\lambda s+\gamma
Y_{s}^{-\xi_T}\right). \label{auxilary_process}
\end{equation}%
We first show that it is a supermartingale for any $\pi \in \mathcal{A}%
_{[t,T]}$, and becomes a martingale for
\begin{equation}
\pi _{s}^{\ast ,\xi_T}=Proj_{\Pi }\left(
Z_{s}^{-\xi_T}+\frac{Z_{s}+\theta (V_{s})}{\gamma }\right) .
\label{aux-optimal}
\end{equation}%
Indeed, for $0\leq t\leq r\leq s\leq T$, observe that the exponent in (\ref%
{auxilary_process}) satisfies%
\begin{eqnarray*}
&&-\gamma ( x+\int_{t}^{s}\pi _{u}^{tr}( \theta( V_{u}) du+dW_{u}) )
+Y_{s}-\lambda s+\gamma
Y_{s}^{-\xi_T}\\
&=&-\gamma ( x+\int_{t}^{r}\pi _{u}^{tr}( \theta ( V_{u})
du+dW_{u})) +Y_{r}-\lambda r+\gamma Y_{r}^{-\xi_T}
\\
&&-\gamma ( x+\int_{r}^{s}\pi _{u}^{tr}( \theta ( V_{u}) du+dW_{u})
) +( Y_{s}-Y_{r}) -\lambda ( s-r) +\gamma (
Y_{s}^{-\xi_T}-Y_{r}^{-\xi_T}) .
\end{eqnarray*}%
Using the ergodic BSDE (\ref{EQBSDE2})\ and the BSDE
(\ref{BSDERepresentation}) yields
\begin{equation*}
\left( Y_{s}-Y_{r}\right) -\lambda \left( s-r\right)
=-\int_{r}^{s}F\left( V_{u},Z_{u}\right)
du+\int_{r}^{s}Z_{u}^{tr}dW_{u}
\end{equation*}%
and
\begin{equation*}
Y_{s}^{-\xi_T}-Y_{r}^{-\xi_T}=-\frac{1}{\gamma }\int_{r}^{s}\left(
F\left( V_{u},Z_{u}+\gamma Z_{u}^{-\xi_T}\right) -F\left(
V_{u},Z_{u}\right) \right) du+\int_{r}^{s}\left(
Z_{u}^{-\xi_T}\right) ^{tr}dW_{u}.
\end{equation*}%
Combining the above gives
\begin{eqnarray*}
&&E_{\mathbb{P}}\left[ \left. -e^{-\gamma \left( x+\int_{t}^{s}\pi
_{u}^{tr}\left( \theta (V_{u})du+dW_{u}\right) \right)
+Y_{s}-\lambda s+\gamma Y_{s}^{-\xi_T}}\right \vert
\mathcal{F}_{r}\right]\\
&=&-e^{-\gamma \left( x+\int_{t}^{r}\pi _{u}^{tr}\left( \theta
(V_{u})du+dW_{u}\right) \right) +Y_{r}-\lambda r+\gamma
Y_{r}^{-\xi_T}}\\
&&\times E_{\mathbb{P}}\left[ \left. e^{\int_{r}^{s}\left( -\gamma \pi _{u}^{{%
tr}}\theta (V_{u})-F(V_{u},Z_{u}+\gamma Z_{u}^{-\xi_T})\right)
du+\int_{r}^{s}\left( -\gamma \pi _{u}+Z_{u}+\gamma Z_{u}^{-\xi_T}\right) ^{{tr}%
}dW_{u}}\right \vert \mathcal{F}_{r}\right] .
\end{eqnarray*}%
For $s\in \lbrack 0,T]$, consider the process
$N_{s}:=\int_{0}^{s}\left( -\gamma \pi _{u}+Z_{u}+\gamma
Z_{u}^{-\xi_T}\right)^{{tr}}dW_{u}$. Because $\pi ,Z^{-\xi_T}\in
\mathcal{L}_{BMO}^{2}[0,T]$ and $Z$ is uniformly bounded, we deduce
that $N$ is a \textit{BMO}-martingale.

Next, we define on $\mathcal{F}_{T}$ a probability measure, denoted by $%
\mathbb{Q}^{\pi },$ by $\frac{d\mathbb{Q}^{\pi }}{d\mathbb{P}}=\mathcal{E}%
(N)_{T}.$ Then, $\left. \frac{d\mathbb{Q}^{\pi }}{d\mathbb{P}}\right \vert _{%
\mathcal{F}_{s}}=\mathcal{E}(N)_{s},$ which is uniformly integrable
due to the \textit{BMO}-martingale property of the process $N$.
Therefore,
\begin{equation*}
e^{\int_{r}^{s}\left( -\gamma \pi _{u}+Z_{u}+\gamma Z_{u}^{-\xi_T}\right) ^{{tr}%
}dW_{u}}=\left( e^{\frac{1}{2}\int_{r}^{s}|-\gamma \pi
_{u}+Z_{u}+\gamma
Z_{u}^{-\xi_T}|^{2}du}\right) \frac{\mathcal{E}(N)_{s}}{\mathcal{E}(N)_{r}}\text{,%
}
\end{equation*}%
and, thus,
\begin{eqnarray*}
&&E_{\mathbb{P}}\left[ \left. -e^{-\gamma \left( x+\int_{t}^{s}\pi
_{u}^{tr}\left( \theta (V_{u})du+dW_{u}\right) \right)
+Y_{s}-\lambda s+\gamma Y_{s}^{-\xi_T}}\right \vert
\mathcal{F}_{r}\right]\\
&=&-e^{-\gamma \left( x+\int_{t}^{r}\pi _{u}^{tr}\left( \theta
(V_{u})du+dW_{u}\right) \right) +Y_{r}-\lambda r+\gamma
Y_{r}^{-\xi_T}}\\
&&\times E_{\mathbb{P}}\left[ e^{\int_{r}^{s}\left( -\gamma \pi _{u}^{{tr}%
}\theta (V_{u})-F(V_{u},Z_{u}+\gamma Z_{u}^{-\xi_T})\right) du+\frac{1}{2}%
\int_{r}^{s}\left|-\gamma \pi _{u}+Z_{u}+\gamma Z_{u}^{-\xi_T}\right|^{2}du}\frac{\mathcal{E}%
(N)_{s}}{\mathcal{E}(N)_{r}}|\mathcal{F}_{r}\right]\\
&=&-e^{-\gamma \left( x+\int_{t}^{r}\pi _{u}^{tr}\left( \theta
(V_{u})du+dW_{u}\right) \right) +Y_{r}-\lambda r+\gamma
Y_{r}^{-\xi_T}}\\
&&\times E_{\mathbb{Q}^{\pi }}\left[ e^{\int_{r}^{s}\left( \left(
-\gamma \pi _{u}^{{tr}}\theta (V_{u})+\frac{1}{2}|-\gamma \pi
_{u}+Z_{u}+\gamma
Z_{u}^{-\xi_T}|^{2}\right) -F(V_{u},Z_{u}+\gamma Z_{u}^{-\xi_T})\right) du}|\mathcal{F}%
_{r}\right] .
\end{eqnarray*}%
In turn, if we can show that, for $u\in \left[ r,s\right] ,$
\begin{equation*}
-\gamma \pi _{u}^{{tr}}\theta (V_{u})+\frac{1}{2}|-\gamma \pi
_{u}+Z_{u}+\gamma Z_{u}^{-\xi_T}|^{2}\geq F(V_{u},Z_{u}+\gamma
Z_{u}^{-\xi_T})\text{,}
\end{equation*}%
then the supermartingality property would follow. Indeed, after some
calculations, we obtain that
\begin{equation*}
-\gamma \pi _{u}^{{tr}}\theta (V_{u})+\frac{1}{2}|-\gamma \pi
_{u}+Z_{u}+\gamma Z_{u}^{-\xi_T}|^{2}
\end{equation*}%
\begin{equation*}
=\frac{\gamma ^{2}}{2}\left \vert \pi _{u}-\left( Z_{u}^{-\xi_T}+\frac{%
Z_{u}+\theta \left( V_{u}\right) }{\gamma }\right) \right \vert ^{2}-\frac{1}{%
2}\left \vert Z_{u}+\gamma Z_{u}^{-\xi_T}+\theta \left( V_{u}\right)
\right \vert ^{2}+\frac{1}{2}\left \vert Z_{u}+\gamma
Z_{u}^{-\xi_T}\right \vert ^{2}.
\end{equation*}%
On the other hand, for any $\pi \in \mathcal{A}_{[t,T]}$,
\begin{equation*}
\left \vert \pi _{u}-\left( Z_{u}^{-\xi_T}+\frac{Z_{u}+\theta (V_{u})}{\gamma }%
\right) \right \vert ^{2}\geq \ {dist}^{2}\left \{ \Pi ,Z_{u}^{-\xi_T}+\frac{%
Z_{u}+\theta (V_{u})}{\gamma }\right \} \text{,}
\end{equation*}%
and using the form of $F(V_{u},Z_{u}+\gamma Z_{u}^{-\xi_T})$ (cf.
(\ref{driver})) we conclude.

To show that $R^{{\pi }^{\ast ,\xi_T}}$ is a martingale for ${\pi
}^{\ast ,\xi_T}$, defined in (\ref{aux-optimal}), observe that
\begin{equation*}
\left \vert {\pi }_{u}^{\ast ,\xi_T}-\left(
Z_{u}^{-\xi_T}+\frac{Z_{u}+\theta (V_{u})}{\gamma }\right) \right
\vert ^{2}=\ {dist}^{2}\left \{ \Pi
,Z_{u}^{-\xi_T}+\frac{Z_{u}+\theta (V_{u})}{\gamma }\right \}
\text{,}
\end{equation*}%
and the martingale property follows. Note, moreover, that {this
policy is admissible. }

We now conclude as follows. Combining the above, we obtain
$E_{\mathbb{P}}\left[ \left. R_{T}^{\pi
}\right \vert \mathcal{F}_{t}\right] \leq R_{t},$ for any $%
\pi \in \mathcal{A}_{[t,T]},$ where we also used that $Y_{T}^{-\xi_T}=-\xi _{T}$ (cf. (\ref{BSDERepresentation}%
)). Similarly, $E_{\mathbb{P}}\left[ \left. R_{T}^{\pi^{*,\xi_T}
}\right \vert \mathcal{F}_{t}\right] = R_{t}.$ In other words,
\begin{equation*}
{u}^{\xi_T}(x,t)=R_t=-e^{-\gamma x+Y_{t}-\lambda t+\gamma
Y_{t}^{-\xi_T}}\text{,}
\end{equation*}%
which, by (\ref{def_1}), implies%
\begin{equation*}
-e^{-\gamma \rho _{t}(\xi _{T})-\gamma x+Y_{t}-\lambda t+\gamma
Y_{t}^{-\xi_T}}=-e^{-\gamma x+Y_{t}-\lambda t}\text{,}
\end{equation*}%
and (\ref{representation}) follows.

\subsection{Proof of Theorem \ref{theorem_dual}}

We start by proving some estimates of the driver $G$ and its convex
dual $G^{*}$.

\begin{lemma} The driver $G(v,z,\bar{z})$ in (\ref{driver_2}) and
its convex dual ${G}^*(v,z,q)$ in (\ref{dual}) have the following
properties:

i) $G(v,z,\bar{z})$ has the upper and lower bounds
\begin{equation}\label{bound_G_1}
-\gamma|\bar{z}|^2-\frac{2}{\gamma}(|z|^2+|\theta(v)|^2)\leq
G(v,z,\bar{z})\leq
\gamma|\bar{z}|^2+\frac{2}{\gamma}(|z|^2+|\theta(v)|^2),
\end{equation}
for any
$(v,z,\bar{z})\in\mathbb{R}^d\times\mathbb{R}^d\times\mathbb{R}^d$.

ii) ${G}^*(v,z,q)$ is a convex function in $q$, for any
$(v,z)\in\mathbb{R}^d\times\mathbb{R}^d$.

iii) ${G}^*(v,z,q)$ has the lower bound
\begin{equation}\label{bound_G_hat}
{G}^*(v,z,q)\geq\max\left\{0,\frac{|q|^2}{4\gamma}-\frac{2}{\gamma}(|z|^2+|\theta(v)|^2)\right\},
\end{equation}
for any
$(v,z,q)\in\mathbb{R}^d\times\mathbb{R}^d\times\mathbb{R}^d$.
\end{lemma}

\begin{proof} The convexity of $G^*(v,z,q)$ in $q$ is immediate, so we
only prove (i) and (iii).

Since $0\in\Pi$, we have that
$dist^{2}\left\{\Pi,\frac{z+\theta(v)}{\gamma}\right\}\leq
\frac{|z+\theta(v)|^2}{\gamma^2}$ and, therefore, $ F(v,z)\leq
\frac12|z|^2$. On the other hand, $F(v,z)\geq
-z^{tr}\theta(v)-\frac12|\theta(v)|^2$. We thus obtain the upper
bound
\begin{eqnarray}\label{bound_G}
G(v,{z},\bar{z})&=&\frac{1}{\gamma }\left( F(v,{z}+\gamma \bar{z})-F(v,{z}%
)\right)\notag\\
&\leq&
\frac{1}{2\gamma}|z+\gamma\bar{z}|^2+\frac{1}{\gamma}(z^{tr}\theta(v)+\frac12|\theta(v)|^2)\notag\\
&\leq&\gamma|\bar{z}|^2+\frac{2}{\gamma}(|z|^2+|\theta(v)|^2),
\end{eqnarray}
which will in turn give us the lower bound of ${G}^*$. Indeed, using
the definition of $G^*$ in (\ref{dual}) and the above upper bound of
$G$ in (\ref{bound_G}), we deduce that
\begin{eqnarray*}
{G}^*(v,z,q)&\geq& \bar{z}^{tr}q-G(v,z,\bar{z})\\
&\geq&
\bar{z}^{tr}q-\gamma|\bar{z}|^2-\frac{2}{\gamma}(|z|^2+|\theta(v)|^2)\\
&=&\frac{|q|^2}{4\gamma}-\frac{2}{\gamma}(|z|^2+|\theta(v)|^2),
\end{eqnarray*}
by taking $\bar{z}=q/2\gamma$. On other hand, since $G(v,z,0)=0$, we
obtain that $G^*(v,z,q)\geq 0$ by taking $\bar{z}=0$.

The lower bound of $G$ is derived in a similar way.\end{proof}
\bigskip

\textbf{Proof of Theorem \ref{theorem_dual}.}

\textbf{Proof of (i)}: For any $q\in\mathcal{A}^*_{[0,T]}$, we
define
\begin{equation*}
Y_t^{-\xi_T,q}:=E_{\mathbb{Q}^q}
\left[\left.-\xi_T-\int_t^{T}{G}^*(V_s,Z_s,q_s)ds\right\vert\mathcal{F}_t\right],
\end{equation*}
which is finite due to the integrability condition on $G^*$ in the
admissible set $\mathcal{A}^*_{[0,T]}$. Note that
$Y_t^{-\xi_T,q}-\int_0^{t}G^{*}(V_s,Z_s,q_s)ds$, $t\in[0,T]$, is a
uniformly integrable martingale under $\mathbb{Q}^q$, so the
martingale representation theorem gives
\begin{equation}\label{BSDE_qq}
Y_t^{-\xi_T,q}-\int_0^{t}G^{*}(V_s,Z_s,q_s)ds=\left(-\xi_T-\int_0^{T}{G}^*(V_s,Z_s,q_s)ds\right)-\int_t^{T}(Z_s^{-\xi_T,q})^{tr}dW_s^{q},
\end{equation}
for some predictable density process $Z^{-\xi_T,q}$, where
$W^{q}_t=W_t-\int_0^{t}q_sds$, $t\in[0,T]$, is a $d$-dimensional
Brownian motion under $\mathbb{Q}^q$.

On the other hand, we rewrite the BSDE (\ref{BSDERepresentation})
under $\mathbb{Q}^q$ as
\begin{equation}\label{BSDE_q}
Y_{t}^{-\xi_T}=-\xi_{T}+\int_{t}^{T}\left(G(V_{s},Z_{s},Z_{s}^{-\xi_T
})-(Z_s^{-\xi_T})^{tr}q_s\right)ds-\int_{t}^{T}(Z_{s}^{-\xi_T})^{tr}dW_{s}^q.
\end{equation}
Combining (\ref{BSDE_qq}) and (\ref{BSDE_q}) and taking the
conditional expectation with respect to $\mathcal{F}_t$ give
\begin{equation*}
Y_{t}^{-\xi_T}-Y_t^{-\xi_T,q}=E_{\mathbb{Q}^q}\left[\left.\int_t^{T}\left(G(V_{s},Z_{s},Z_{s}^{-\xi_T
})-(Z_s^{-\xi_T})^{tr}q_s+{G}^*(V_s,Z_s,q_s)\right)ds\right\vert\mathcal{F}_t\right].
\end{equation*}
Using (\ref{FM_theorem}), we then deduce that, for any
$q\in\mathcal{A}^*_{[0,T]}$,
\begin{equation*}
G(V_{s},Z_{s},Z_{s}^{-\xi_T
})-(Z_s^{-\xi_T})^{tr}q_s+{G}^*(V_s,Z_s,q_s)\geq 0,
\end{equation*}
and thus $Y_t^{-\xi_T}\geq Y_t^{-\xi_T,q}$. Setting $q^*_s:=\partial
G_{\bar{z}}(V_s,Z_s,Z_s^{-\xi_T})$, we further obtain from
(\ref{FM_theorem_2}) that
\begin{equation}\label{dual_formula_3}
G(V_{s},Z_{s},Z_{s}^{-\xi_T
})-(Z_s^{-\xi_T})^{tr}q_s^*+{G}^*(V_s,Z_s,q_s^*)=0,
\end{equation}
from which we conclude that $Y_t^{-\xi_T}= Y_t^{-\xi_T,q^*}$, for
$t\in[0,T]$.

\textbf{Proof of (ii)}: We show that the above $q^*_s$, $s\in[0,T]$,
is in the admissible set $\mathcal{A}^{*}_{[0,T]}$. To this end,
using the lower bound of $G^*$ in (\ref{bound_G_hat}), we deduce
from (\ref{dual_formula_3}) that
\begin{eqnarray*}
G(V_{s},Z_{s},Z_{s}^{-\xi_T
})&=&(Z_s^{-\xi_T})^{tr}q_s^*-{G}^*(V_s,Z_s,q_s^*)\\
&\leq& (Z_s^{-\xi_T})^{tr}q_s^*-
\frac{|q_s^*|^2}{4\gamma}+\frac{2}{\gamma}(|Z_s|^2+|\theta(V_s)|^2)\\
&\leq&2\gamma|Z_s^{-\xi_T}|^2+\frac{|q_s^*|^2}{8\gamma}-
\frac{|q_s^*|^2}{4\gamma}+\frac{2}{\gamma}(|Z_s|^2+|\theta(V_s)|^2),
\end{eqnarray*}
where we used the elementary inequality $ab\leq
2\gamma|a|^2+\frac{|b|^2}{8\gamma}$ in the last inequality.
Combining the above inequality and the lower bound of $G$ in
(\ref{bound_G_1}) further yields that
\begin{eqnarray*}
\frac{1}{8\gamma}|q_s^*|^2&\leq&
2\gamma|Z_s^{-\xi_T}|^2+\frac{2}{\gamma}(|Z_s|^2+|\theta(V_s)|^2)-G(V_s,Z_s,Z_s^{-\xi_T})\\
&\leq& 3\gamma
|Z_s^{-\xi_T}|^2+\frac{4}{\gamma}(|Z_s|^2+|\theta(V_s)|^2).
\end{eqnarray*}
Since $Z^{-\xi_T}\in\mathcal{L}^2_{BMO}[0,T]$, and both $Z$ and
$\theta(V)$ are bounded, we obtain that
$q^*\in\mathcal{L}^2_{BMO}[0,T]$.

Finally, using (\ref{dual_formula_3}) and the bounds of $G$ in
(\ref{bound_G_1}), we deduce that
\begin{eqnarray*}
&&E_{\mathbb{Q}^{q^*}}\left[\int_0^T|G^*(V_s,Z_s,q_s^*)|ds\right]\\
&=&E_{\mathbb{Q}^{q^*}}\left[\int_0^T|(Z_s^{-\xi_T})^{tr}q_s^{*}-G(V_s,Z_s,Z_s^{-\xi_T})|ds\right]\\
&\leq&(\frac12+\gamma)E_{\mathbb{Q}^{q^*}}\left[\int_0^T|Z_s^{-\xi_T}|^2ds\right]+\frac12E_{\mathbb{Q}^{q^*}}\left[\int_0^T|q_s^*|^2ds\right]\\
&&+\frac{2}{\gamma}E_{\mathbb{Q}^{q^*}}\left[\int_0^T|Z_s|^2+|\theta(V_s)|^2ds\right].
\end{eqnarray*}
Since $Z^{-\xi_T},q^*\in\mathcal{L}^{2}_{BMO}[0,T]$ under
$\mathbb{Q}$, and $\mathbb{Q}\sim\mathbb{Q}^{q^*}$, we obtain that
$Z^{-\xi_T},$ $q^*\in\mathcal{L}^{2}_{BMO}[0,T]$ under
$\mathbb{Q}^{q^*}$ (see, for example, pp. 1563 in \cite{XYZ}) and,
therefore,
$E_{\mathbb{Q}^{q^*}}[\int_0^T|G^*(V_s,Z_s,q_s^*)|ds]<\infty$. We then
conclude that $q^*\in\mathcal{A}^*_{[0,T]}$.

\subsection{Proof of Theorem \ref{thm:long_term_price}}

As we explained in Remark \ref{Remark}, it is difficult to analyze
the ergodic BSDE (\ref{BSDERepresentation 22}). We will instead work
with the function $y^{T,g}(\cdot,\cdot)$ directly.

We first establish some auxiliary estimates.

\begin{lemma}
\label{lemma_11} Suppose that Assumptions \ref{assumption1} and \ref%
{assumption2} hold, and the risk position $\xi _{T}$ is as in
(\ref{payoff-european}). Then, the function $y^{T,g}(v,t)$,
$(v,t)\in\mathbb{R}^d\times[0,T]$, has the following properties.

i) There exists a constant $C>0$ such that
\begin{equation*}
|y^{T,g}(v,t)|\leq C(1+|v|).
\end{equation*}

ii) With the constant $q$ given in (\ref{q-constant}),
\begin{equation*}
|\nabla y^{T,g}(v,t)|\leq q+\frac{C_{v}}{\gamma (C_{\eta }-C_{v})}.
\end{equation*}

iii) With the constant $\hat{C}_{\eta }$ given in Proposition \ref%
{prop:forward_estimate},
\begin{equation*}
|y^{T,g}(v,t)-y^{T,g}(\bar{v},t)|\leq C(1+|v|^{2}+|%
\bar{v}|^{2})e^{-\hat{C}_{\eta }(T-t)}.
\end{equation*}
\end{lemma}

\begin{proof} Fixing $t\in[0,T]$, and for the stochastic factor process starting from $V_t^{t,v}=v$, we recall that $(Y_s,Z_s)=(y(V_s^{t,v}),z(V_s^{t,v}))$ and that
\begin{equation*}
(Y_{s}^{-\xi_T},Z_{s}^{-\xi_T})=(y^{T,g}(V_s^{t,v},s),z^{T,g}(V_s^{t,v},s))\
\text{for}\ s\in[t,T],
\end{equation*}
where
\begin{eqnarray}
Y_{s}^{-\xi_T}&=&g(V_{T}^{t,v})+\int_{s}^{T}\frac{1}{\gamma }\left(
F(V_{u}^{t,v},\gamma
\hat{z}(V_u^{t,v},u))-F(V_{u}^{t,v},Z_{u})\right)
du  \notag \\
&&-\int_{s}^{T}\left( \hat{z}(V_u^{t,v},u)-\frac{Z_{u}}{\gamma }%
\right) ^{tr}dW_{u}\text{,}  \label{BSDERepresentation_1}
\end{eqnarray}%
with $\hat{z}(V_s^{t,v},s)$ defined as
\begin{equation}\label{hatZ}
\hat{z}(V_s^{t,v},s):=Z_s^{-\xi_T}+\frac{Z_s}{\gamma}=z^{T,g}(V_s^{t,v},s)+\frac{z(V_s^{t,v})}{\gamma}.
\end{equation}
In Lemma \ref{lemma} of Appendix A.2, we will prove that
$|\hat{z}(\cdot,\cdot)|\leq q$. Thus, the process $Z^{-\xi_T}$ is
uniformly bounded, since
\begin{equation*}
|Z_{s}^{-\xi_T}|=|\hat{z}(V_s^{t,v},s)-\frac{Z_{s}}{\gamma }|\leq q+%
\frac{C_{v}}{\gamma (C_{\eta }-C_{v})}\text{.}
\end{equation*}%
Therefore, the gradient estimate (ii) for $y^{T,g}(v,t)$ follows
from the relationship $\kappa ^{tr}\nabla y^{T,g}(V_{s}^{t,v},s)=
Z_{s}^{-\xi_T}$.

To prove (i), we introduce
\begin{equation}
H(V_{s}^{t,v}):=\frac{\left( F(V_{s}^{t,v},\gamma \hat{z}(V_s^{t,v},s))-F(V_{s}^{t,v},Z_{s})\right) \left( \hat{z}(V_s^{t,v},s)-%
\frac{Z_{s}}{\gamma }\right) }{\gamma \left \vert
\hat{z}(V_s^{t,v},s)-\frac{Z_{s}}{\gamma} \right \vert
^{2}}\mathbf{1}_{\left \{ \hat{z}(V_s^{t,v},s)-\frac{Z_{s}}{\gamma
}\neq 0\right \} }\text{,} \label{term_m}
\end{equation}%
and observe that it is uniformly bounded due to (\ref{driver1})\textbf{\ }%
and the boundedness of $\hat{z}(\cdot,\cdot)$ and $Z$.

Next, define
a probability measure $%
\mathbb{Q}^{H}$ by $\frac{d\mathbb{Q}^{H}}{d\mathbb{P}}:=\mathcal{E}%
(\int_t^{\cdot}(H(V_s^{t,v}))^{tr}dW_s)_{T}$ on $\mathcal{F}_{T}$.
Then, equation (\ref%
{BSDERepresentation_1}) can be written as
\begin{equation*}
Y_{t}^{-\xi_T}=y^{T,g}(v,t)
\end{equation*}%
\begin{equation*}
=g(V_{T}^{t,v})-\int_{t}^{T}\left( \hat{z}(V_s^{t,v},s)-\frac{Z_{s}}{%
\gamma }\right) ^{tr}(dW_{s}-H(V_{s}^{t,v})ds)=E_{\mathbb{Q}%
^{H}}[g(V_{T}^{t,v})|\mathcal{F}_{t}],
\end{equation*}%
and the assertion follows from the linear growth property of
$g(\cdot )$ and the first assertion of part (ii) in Proposition
\ref{prop:forward_estimate} of Appendix A.1.

Finally, for $v,\bar{v}\in \mathbb{R}^{d}$, by the second assertion
of (ii) in Proposition \ref{prop:forward_estimate}, we deduce that
\begin{equation*}
|y^{T,g}(v,t)-y^{T,g}(\bar{v},t)|=\left \vert E_{\mathbb{Q}%
^{H}}[g(V_{T}^{t,v})-g(V_{T}^{t,\bar{v}})|\mathcal{F}_{t}]\right
\vert
\end{equation*}%
\begin{equation*}
=\left \vert E_{\mathbb{Q}^{H}}[g(V_{T-t}^{0,v})-g(V_{T-t}^{0,\bar{v}%
})]\right \vert \leq C\left( 1+|v|^{2}+|\bar{v}|^{2}\right) e^{-\hat{C}%
_{\eta }(T-t)},
\end{equation*}
and we conclude.
\end{proof}
\bigskip

\textbf{Proof of Theorem \ref{thm:long_term_price}.}

\textbf{Proof of (i)}: From the first estimate (i) in Lemma \ref%
{lemma_11}, we first construct, using a standard diagonal procedure,
a
sequence $\{T_{i}\}_{i=1}^{\infty }$ such that $T_{i}\uparrow \infty $, and $%
\lim_{T_{i}\uparrow \infty }y^{T_{i},g}\left( v,0\right) =$
$L^{g}\left( v\right) ,$ for $v\in D$, where $D$ is a dense subset
of $\mathbb{R}^{d}$, and some function $L^{g}\left( v\right) .$

Moreover, the second estimate (ii) in Lemma \ref{lemma_11} implies
that, for $v,\bar{v}\in \mathbb{R}^{d}$,
\begin{equation}
|y^{T,g}\left( v,0\right) -y^{T,g}\left( \bar{v},0\right) |\leq \left( q+\frac{%
C_{v}}{\gamma (C_{\eta }-C_{v})}\right) |v-\bar{v}|.
\label{eq:2nd_estimate}
\end{equation}%
Therefore, the limit $L^{g}(v)$ can be extended to a Lipschitz
continuous function, defined for all $v\in \mathbb{R}^{d}$, and,
furthermore, we have that
\begin{equation*}
\lim_{T_{i}\uparrow \infty }y^{T_{i},g}(v,0)=L^{g}(v),\text{\  \ }v\in \mathbb{R%
}^{d}\text{.}
\end{equation*}%
Indeed, for $v\in \mathbb{R}^{d}\backslash D$, there exists a sequence $%
\{v_{j}\}_{j=1}^{\infty }\subseteq D$ such that $v_{j}\rightarrow
v$. Define $L^{g}\left( v\right) :=\lim_{j\uparrow \infty
}L^{g}\left( v_{j}\right) $. Using the estimate
(\ref{eq:2nd_estimate}), we have
\begin{equation*}
|y^{T_{i},g}\left( v,0\right) -y^{T_{i},g}\left( v_{j},0\right) |\leq \left( q+%
\frac{C_{v}}{\gamma (C_{\eta }-C_{v})}\right) |v-v_{j}|\text{.}
\end{equation*}%
Taking $T_{i}\uparrow \infty $ and since $\lim_{j\uparrow \infty
}y^{T_{i},g}\left( v_{j},0\right) =L^{g}\left( v_{j}\right) $, we
obtain
\begin{equation*}
\left \vert \lim_{T_{i}\uparrow \infty }y^{T_{i},g}\left( v,0\right)
-L^{g}\left( v_{j}\right) \right \vert \leq \left(
q+\frac{C_{v}}{\gamma (C_{\eta }-C_{v})}\right) |v-v_{j}|\text{.}
\end{equation*}%
Sending $j\uparrow \infty $, we deduce that, for any $v\in \mathbb{R}^{d}$, $%
\lim_{T_{i}\uparrow \infty }y^{T_{i},g}\left( v,0\right)
=L^{g}\left( v\right) \text{.}$

Next, we show that, for $v\in\mathbb{R}^d$, $L^{g}(v)\equiv L^{g},$
a constant function. To this end,
by the third estimate (iii) in Lemma \ref{lemma_11}, we have, for any $v,%
\bar{v}\in \mathbb{R}^{d}$, that
\begin{equation*}
|y^{T_{i},g}\left( v,0\right) -y^{T_{i},g}\left( \bar{v},0\right)
|\leq C\left( 1+|v|^{2}+|\bar{v}|^{2}\right) e^{-\hat{C}_{\eta
}T_{i}}.
\end{equation*}%
Letting $T_{i}\uparrow \infty $ yields $\lim_{T_{i}\uparrow \infty
}y^{T_{i},g}\left( v,0\right) =\lim_{T_{i}\uparrow \infty
}y^{T_{i},g}\left( \bar{v},0\right) $, which implies that the limit
function $L^{g}(v)$ is independent of $v$. Thus, it is a constant,
denoted as $L^{g}$. Moreover, such a
constant $L^{g}$ is independent of the choice of the sequence $%
\{T_{i}\}_{i=1}^{\infty }$ (see pp. 394-395 in \cite{Hu11} for its
proof).

To prove the convergence rate (\ref{convergence_rate}), we argue as
follows.
For $v\in \mathbb{R}^{d}$ and $T>0$, we have, from the proof of Lemma \ref%
{lemma_11} (i), that%
\begin{eqnarray*}
|y^{T,g}\left( v,0\right) -L^{g}| &=&\  \lim_{T^{\prime }\uparrow
\infty }|y^{T,g}\left( v,0\right) -y^{T^{\prime},g}\left( v,0\right)
|\\
&=&\  \lim_{T^{\prime }\uparrow \infty }\left \vert y^{T,g}\left( v,0\right) -E_{%
\mathbb{Q}^{H}}\left[ g(V_{T^{\prime }}^{0,v})\right] \right \vert.
\end{eqnarray*}
From the tower property of conditional expectations, we further
deduce, for $T^{\prime}>T$, that
\begin{eqnarray*}
\left \vert y^{T,g}\left( v,0\right) -E_{%
\mathbb{Q}^{H}}\left[ g(V_{T^{\prime }}^{0,v})\right] \right \vert&=&\left \vert y^{T,g}\left( v,0\right) -E_{%
\mathbb{Q}^{H}}\left[ E_{\mathbb{Q}^{H}}\left[ g(V_{T^{\prime }}^{0,v})|%
\mathcal{F}_{T^{\prime}-T}\right] \right] \right \vert\\
&=&\left \vert y^{T,g}\left( v,0\right) -E_{%
\mathbb{Q}^{H}}\left[ y^{T^{\prime},g}\left( V_{T^{\prime
}-T}^{0,v},T^{\prime }-T\right) \right] \right \vert
\\
&=&\left \vert y^{T,g}\left( v,0\right) -E_{%
\mathbb{Q}^{H}}\left[ y^{T,g}\left( V_{T^{\prime }-T}^{0,v},0\right)
\right] \right \vert\\
&=&\left \vert E_{\mathbb{Q}^{H}}\left[ y^{T,g}\left( v,0\right)
-y^{T,g}\left( V_{T^{\prime }-T}^{0,v},0\right) \right] \right
\vert.
\end{eqnarray*}
Therefore,
\begin{eqnarray*}
|y^{T,g}\left( v,0\right) -L^{g}| &=&\  \lim_{T^{\prime }\uparrow
\infty }\left \vert E_{\mathbb{Q}^{H}}\left[ y^{T,g}\left(
v,0\right)
-y^{T,g}\left( V_{T^{\prime }-T}^{0,v},0\right) \right] \right \vert\\
&\leq& \  \lim_{T^{\prime }\uparrow \infty
}CE_{\mathbb{Q}^{H}}\left[
1+|v|^{2}+\left \vert V_{T^{\prime }-T}^{0,v}\right \vert ^{2}\right] e^{-%
\hat{C}_{\eta }T}\\
&\leq& \ C\left( 1+|v|^{2}\right) e^{-\hat{C}_{\eta }T}\text{,%
}
\end{eqnarray*}%
where we used (ii) in Proposition \ref{prop:forward_estimate} and
(iii) in Lemma \ref{lemma_11} in the last two inequalities.

\textbf{Proof of (ii)}: We only establish the exponential bound of
the hedging strategy $\alpha _{t,T}$ in (\ref{exponential_bound}).
Then, the asymptotic behavior of $\alpha _{t,T}$ in
(\ref{hedging_strategy}) follows by letting $T\uparrow \infty .$

From Remark \ref{remark_strategy} and the Lipschitz continuity of
the projection operator on the convex set $\Pi $, we deduce that,
for any $s\in[0,T)$,
\begin{eqnarray*}
&&E_{\mathbb{P}}\left[\int_0^{s}|\alpha_{u,T}|^2du\right]\\
&=&
E_{\mathbb{P}}\left[\int_0^{s}\left|Proj_{\Pi}\left(Z_u^{-\xi_T}+\frac{Z_u+\theta(V_u)}{\gamma}\right)-
Proj_{\Pi}\left(\frac{Z_u+\theta(V_u)}{\gamma}\right)\right|^2du\right]\\
&\leq& C E_{%
\mathbb{P}}\left[ \int_{0}^{s}|Z_u^{-\xi_T}|^{2}du\right].
\end{eqnarray*}
Thus, we only need to establish the exponential bound of
$Z^{-\xi_T}_{u}=z^{T,g}(V_u^v,u)$ with the stochastic factor process
starting from $V_0^{v}=v$. To this end, we easily deduce, using
(iii) in Lemma \ref{lemma_11}, that, for $t\in \lbrack 0,T)$,
\begin{equation}
|y^{T,g}(v,t)-L^{g}|\leq C(1+|v|^{2})e^{-\hat{C}_{\eta }(T-t)}.
\label{convergence_rate_11}
\end{equation}%
Applying It{\^{o}}'s formula to $|y^{T,g}(V_{s}^{v},s)-L^{g}|^{2}$ and using (\ref%
{BSDERepresentation_1}), we in turn have%
\begin{eqnarray*}
&&|y^{T,g}(v,0)-L^{g}|^{2}+E_{\mathbb{P}}\left[ \int_{0}^{s}|Z_u^{-\xi_T}|^{2}du%
\right]\\
&=&
E_{\mathbb{P}}[|y^{T,g}(V_{s}^{v},s)-L^{g}|^{2}]\\
&&+2E_{\mathbb{P}}\left[
\int_{0}^{s}|y^{T,g}(V_{u}^v,u)-L^{g}|\frac{F(V_{u}^v,\gamma \hat{z}(V_u^{v},u))-F(V_{u}^v,Z_u)}{\gamma }du\right]\\
&=&
E_{\mathbb{P}}[|y^{T,g}(V_{s}^v,s)-L^{g}|^{2}]+2E_{\mathbb{P}}\left[
\int_{0}^{s}(Z^{-\xi_T}_u)^{tr}H(V_{u}^v)|y^{T,g}(V_{u}^v,u)-L^{g}|du\right],
\end{eqnarray*}%
where $\hat{z}(\cdot,\cdot)$ is given in (\ref{hatZ}), and the
process $H(V_u^{v})$, introduced in (\ref{term_m}), is uniformly
bounded. Using the elementary inequality $ab\leq
\frac{1}{4}|a|^2+|b|^2$, we further obtain that
\begin{eqnarray*}
&&E_{\mathbb{P}}\left[
\int_{0}^{s}(Z_u^{-\xi_T})^{tr}H(V_{u}^v)|y^{T,g}(V_{u}^v,u)-L^{g}|du\right]\\
&\leq &\frac{1}{4}E_{\mathbb{P}}%
\left[ \int_{0}^{s}|Z_u^{-\xi_T}|^{2}du\right]
+CE_{\mathbb{P}}\left[
\int_{0}^{s}|y^{T,g}(V_{u}^v,u)-L^{g}|^{2}du\right].
\end{eqnarray*}
Hence, (\ref{convergence_rate_11}) yields that
\begin{eqnarray*}
&&\frac12E_{\mathbb{P}}\left[ \int_{0}^{s}|Z_u^{-\xi_T}|^{2}du\right]\\
&\leq&
E_{\mathbb{P}}[|y^{T,g}(V_{s}^v,s)-L^{g}|^{2}]+CE_{\mathbb{P}}\left[
\int_{0}^{s}|y^{T,g}(V_{u}^v,u)-L^{g}|^{2}du\right]\\
&\leq& Ce^{-2\hat{C}%
_{\eta }T}\left( e^{2\hat{C}_{\eta }s}E_{\mathbb{P}}[(1+|V_{s}^v|^{2})^{2}]+%
\int_{0}^{s}e^{2\hat{C}_{\eta }u}E_{\mathbb{P}}[(1+|V_{u}^v|^{2})^{2}]du%
\right),
\end{eqnarray*}%
from which we conclude that $E_{\mathbb{P}}\left[
\int_{0}^{s}|Z_u^{-\xi_T}|^{2}du\right]\leq
C(1+|v|^4)e^{-2\hat{C}_{\eta}(T-s)}$, using the first assertion of
part (ii) in Proposition \ref{prop:forward_estimate}.


\section{A parity result between forward and classical entropic risk measures%
}

\label{section:classical} In this section, we relate forward
entropic risk measures to their classical analogue. In the latter
case, the investment horizon is finite, say $\left[ 0,T\right] $,
for some fixed $T,$ and the
terminal utility is given by%
\begin{equation}
U_{T}(x)=-e^{-\gamma x},  \label{exponential-classical}
\end{equation}%
$x\in \mathbb{R},$ $\gamma >0.$ We recall the definition of
classical entropic risk measures associated with this utility (see,
among others, \cite{Bec0, ElKaroui_2000, Henderson,HH-Carmona,
Henderson3, HU0,MS,MZ4, MR2124276}).

\begin{definition}
\label{def_classical_risk_measure} Let $T>0$ be fixed, and consider
a risk position introduced at $t=0,$ yielding payoff $\xi _{T}\in
\mathcal{F}_{T}.$ Its entropic risk measure, denoted by $\rho
_{t,T}(\xi _{T})\in \mathcal{F}_{t}$, is defined by
\begin{equation}
p^{\xi_T}(x+\rho_{t,T}(\xi_T),t)=p^{0}(x,t) \label{def_2},
\end{equation}%
for all $\left( x,t\right) \in \mathbb{R}\times[0,T]$, with $U_{T}(\cdot )$ as in (\ref%
{exponential-classical}), and
\begin{equation*}
p^{\xi_T}(x,t):=\esssup_{\pi \in
\mathcal{A}_{[t,T]}}E_{\mathbb{P}}\left[ \left.
U_{T}(x+\int_{t}^{T}\pi _{s}^{tr}\left( \theta
(V_{s})ds+dW_{s}\right)+\xi_T)\right \vert \mathcal{F}_{t}\right],
\end{equation*}
provided the above essential supremum exists.
\end{definition}

Next, we present a decomposition formula that relates forward
entropic risk measures under the exponential performance criterion
with their classical counterpart.

\begin{proposition}
\label{proposition14} Suppose that Assumptions 1 and 2 hold, and let
$\xi _{T}$ be a risk position as in (\ref{payoff-european}). Then,
for $t\in \left[ 0,T\right] ,$ the forward $\rho _{t}(\cdot )$ and
classical $\rho _{t,T}\left( \cdot \right) $ entropic risk measures
satisfy
\begin{equation}
\rho _{t}(\xi _{T})=\rho _{t,T}(\xi _{T}-\frac{Y_{T}-\lambda T}{\gamma }%
)-\rho _{t,T}(-\frac{Y_{T}-\lambda T}{\gamma })\text{,}
\label{decomposition}
\end{equation}%
where $\left( Y,\lambda \right) $ is the unique Markovian solution
to the ergodic BSDE (\ref{EQBSDE2}).
\end{proposition}

\begin{proof}
Let $\bar{\xi}_{T}:=\xi _{T}-\frac{Y_{T}-\lambda T}{\gamma
}=-g(V_T^v)-\frac{y(V_T^v)-\lambda T}{\gamma}$. Using arguments
similar to the ones used in section 3 of \cite{HU0}, we deduce that
\begin{equation*}
p^{\bar{\xi}_T}(x+\rho_{t,T}(\bar{\xi}_T),t)=-e^{-\gamma x-\gamma \rho _{t,T}(\bar{\xi}_{T})+\gamma P_{t}^{-\bar{\xi}%
_{T}}},
\end{equation*}%
where the process $P^{-\bar{\xi}_{T}}_t$, $t\in[0,T]$, solves the
quadratic BSDE
\begin{equation}
P_{t}^{-\bar{\xi}_{T}}=-\bar{\xi}_{T}+\int_{t}^{T}\frac{1}{\gamma
}F\left( V_{s}^v,\gamma Q_{s}^{-\bar{\xi}_{T}}\right)
ds-\int_{t}^{T}\left( Q_{s}^{-\bar{\xi}_{T}}\right)
^{tr}dW_{s}\text{,}  \label{BSDE_11}
\end{equation}%
with $F(\cdot ,\cdot )$ as in (\ref{driver}).

The above BSDE is the same as the BSDE (\ref{Auxilary_BSDE}), which
admits a unique solution $(P^{-\bar{\xi}_T},Q^{-\bar{\xi}_T})$, as
stated in Lemma \ref{lemma} of Appendix A.2.

Next, from (\ref{def_2}) in Definition
\ref{def_classical_risk_measure}, we obtain that
\begin{equation}
\rho _{t,T}(\bar{\xi}_{T})=P_{t}^{-\bar{\xi}_{T}}-P_{t}^{0}\text{.}
\label{equ10}
\end{equation}%
By Theorem \ref{theorem}, we have that $\rho _{t}(\xi
_{T})=Y_{t}^{-\xi_T}$, and in turn, that
\begin{equation}
\rho _{t}(\xi _{T})=P_{t}^{-\bar{\xi}_{T}}-\frac{Y_{t}-\lambda t}{\gamma }%
\text{.}  \label{equ11}
\end{equation}%
Taking $\xi _{T}\equiv 0$ yields $\rho _{t}(0)=P_{t}^{\frac{%
Y_{T}-\lambda T}{\gamma }}-\frac{Y_{t}-\lambda t}{\gamma },$ and
since $\rho _{t}(0)=0$, we obtain that
\begin{equation}
P_{t}^{\frac{Y_{T}-\lambda T}{\gamma }}=\frac{Y_{t}-\lambda t}{\gamma }%
\text{.}  \label{equ12}
\end{equation}%
Therefore, (\ref{equ11}) and (\ref{equ12}) yield
\begin{eqnarray*}
\rho _{t}(\xi _{T})&=&P_{t}^{-\bar{\xi}_{T}}-P_{t}^{\frac{Y_{T}-\lambda T}{%
\gamma }}\\
&=&P_{t}^{-\bar{\xi}_{T}}-P_{t}^{0}-(P_{t}^{\frac{Y_{T}-\lambda T}{%
\gamma }}-P_{t}^{0})\\
&=&\rho _{t,T}(\bar{\xi}_{T})-\rho _{t,T}(-\frac{Y_{T}-\lambda
T}{\gamma }),
\end{eqnarray*}%
where we used (\ref{equ10}) in the last equality.
\end{proof}

\section{An example}

We provide an example in which we derive explicit formulae for both
the forward and classical entropic risk measures. We also provide
numerical results for their long-maturity limits.

To this end, we consider a market with a single stock whose
coefficients depend on a single stochastic factor driven by a
2-dimensional Brownian motion,
namely,%
\begin{equation*}
dS_{t}=b\left( V_{t}\right) S_{t}dt+\sigma \left( V_{t}\right)
S_{t}dW_{t}^{1},
\end{equation*}%
\begin{equation*}
dV_{t}=\eta \left( V_{t}\right) dt+\kappa _{1}dW_{t}^{1}+\kappa
_{2}dW_{t}^{2},
\end{equation*}%
for positive constants $\kappa _{1},\kappa _{2}.$

We assume that $|\kappa _{1}|^{2}+|\kappa _{2}|^{2}=1$, the functions $%
b(\cdot )$ and $\sigma (\cdot)$ are uniformly bounded with
$\sigma(\cdot)>0$, and $\eta
(\cdot )$ satisfies the dissipative condition in Assumption \ref{assumption2}%
. We also choose $\Pi =\mathbb{R}\times \{0\},$ so that $\pi _{2,t}\equiv 0$%
.

Then, the wealth equation (\ref{wealth}) becomes $dX_{t}^{\pi
_{1}}=\pi _{1,t}\left( \theta (V_{t})dt+dW_{t}^{1}\right) ,$ where
$\theta \left( V_{t}\right) =\frac{b\left( V_{t}\right) }{\sigma
\left( V_{t}\right) }$. The risk position is given by
$\xi_T=-g(V_T)$, as in (\ref{payoff-european}).

\subsection{Forward entropic risk measures}

The drivers in (\ref{driver}) and (\ref%
{driver_2}) become
\begin{eqnarray*}
F\left( v,(z_{1},z_{2})^{tr}\right) &=&-\frac{1}{2}|z_{1}+\theta
\left( v\right)
|^{2}+\frac{1}{2}|z_{1}|^{2}+\frac{1}{2}|z_{2}|^{2}\\
&=&-\theta \left( v\right) z_{1}-\frac{1}{2}|\theta \left( v\right) |^{2}+%
\frac{1}{2}|z_{2}|^{2}\text{,}
\end{eqnarray*}%
and
\begin{eqnarray*}
G\left( v,(z_{1},z_{2})^{tr},(\bar{z}_{1},\bar{z}_{2})^{tr}\right) &=&\frac{1}{%
\gamma }\left( F(v,(z_{1}+\gamma \bar{z}_{1},z_{2}+\gamma \bar{z}%
_{2})^{tr})-F(v,(z_{1},z_{2})^{tr})\right)\\
&=&-\theta \left( v\right) \bar{z}_{1}+z_{2}\bar{z}_{2}+\frac{\gamma }{2}|\bar{%
z}_{2}|^{2}\text{,}
\end{eqnarray*}%
for $z=(z_{1},z_{2})^{tr}\in\mathbb{R}^2$ and
$\bar{z}=(\bar{z}_{1},\bar{z}_{2})^{tr}\in\mathbb{R}^2.$

Then, the convex dual of $G$ is given as
\begin{equation*}
G^*\left(v,(z_1,z_2)^{tr},(q_1,q_2)^{tr}\right)=
\frac{|q_2-z_2|^2}{2\gamma}\times\mathbf{1}_{\{q_1+\theta(v)=0\}}+\infty\times\mathbf{1}_{\{q_1+\theta(v)\neq
0\}}.
\end{equation*}
for $q=(q_1,q_2)^{tr}\in\mathbb{R}^2$.

To obtain the explicit solution, we set $Z_{1,t}^{-\xi_T}:=\kappa
_{1}Z_{t}^{-\xi_T}$ and $Z_{2,t}^{-\xi_T}:=\kappa
_{2}Z_{t}^{-\xi_T}$ for some process $Z_{t}^{-\xi_T}$ to be
determined. Then, equation
(\ref{BSDERepresentation}) becomes%
\begin{eqnarray*}
dY_{t}^{-\xi_T}&=&-\left( \left( -\kappa _{1}\theta \left(
V_{t}\right) +\kappa
_{2}Z_{2,t}\right) Z_{t}^{-\xi_T}+\frac{\gamma |\kappa _{2}|^{2}}{2}%
|Z_{t}^{-\xi_T}|^{2}\right) dt\\
&&+Z_{t}^{-\xi_T}\left( \kappa _{1}dW_{t}^{1}+\kappa
_{2}dW_{t}^{2}\right) ,
\end{eqnarray*}%
with $Y_{T}^{-\xi_T}=-\xi_T=g\left( V_{T}\right).$

We define $\tilde{Y}_{t}^{-\xi_T}:=e^{\gamma
|\kappa _{2}|^{2}Y_{t}^{-\xi_T}}$ and $\tilde{Z}_{t}^{-\xi_T}:=\gamma |\kappa _{2}|^{2}%
\tilde{Y}_{t}^{-\xi_T}Z_{t}^{-\xi_T}$. In turn,%
\begin{equation*}
d\tilde{Y}_{t}^{-\xi_T}=\tilde{Z}_{t}^{-\xi_T}\left( \left( \kappa
_{1}\theta \left( V_{t}\right) -\kappa _{2}Z_{2,t}\right) dt+\kappa
_{1}dW_{t}^{1}+\kappa _{2}dW_{t}^{2}\right) ,
\end{equation*}%
with $\tilde{Y}_{T}^{-\xi_T}=e^{\gamma |\kappa _{2}|^{2}g\left(
V_{T}\right) }.$ Since $\theta (\cdot )$ and $Z_{2}$ are uniformly
bounded, the process
\begin{equation*}
B_{t}:=\int_{0}^{t}\left( \kappa _{1}\theta \left( V_{s}\right)
-\kappa _{2}Z_{2,s}\right) ds+\kappa _{1}dW_{s}^{1}+\kappa
_{2}dW_{s}^{2},
\end{equation*}%
$t\geq 0,$ is a Brownian motion under some measure $\mathbb{Q}$,
equivalent to $\mathbb{P}$, with
\begin{equation}
\left. \frac{d\mathbb{Q}}{d\mathbb{P}}\right \vert _{\mathcal{F}_{t}}:=%
\mathcal{E}\left( -\int_{0}^{\cdot }\left( \kappa _{1}\theta \left(
V_{s}\right) -\kappa _{2}Z_{2,s}\right) (\kappa
_{1}dW_{s}^{1}+\kappa _{2}dW_{s}^{2})\right) _{t}\text{.}
\label{Q-forward}
\end{equation}%
Hence, $d\tilde{Y}_{t}^{-\xi_T}=\tilde{Z}_{t}^{-\xi_T}dB_t$ and, thus, $%
\tilde{Y}_{t}^{-\xi_T}=E_{\mathbb{Q}}\left[ e^{\gamma |\kappa
_{2}|^{2}g\left( V_{T}\right) }|\mathcal{F}_{t}\right] $.

In turn, the forward entropic risk measure has the closed-form
representation
\begin{equation}
\rho _{t}(\xi _{T})=Y_{t}^{-\xi_T}=
\frac{1}{\gamma |\kappa _{2}|^{2}}\ln E_{%
\mathbb{Q}}\left[ e^{\gamma |\kappa _{2}|^{2}g\left( V_{T}\right) }|\mathcal{%
F}_{t}\right] \text{.}  \label{closed-form}
\end{equation}

Its convex dual representation is then given by
\begin{eqnarray}\label{convex_formula }
\rho_t(\xi_T)&=&-\essinf_{q\in\mathcal{A}^*_{[t,T]}}E_{\mathbb{Q}^q}
\left[\left.\xi_T+\int_t^{T}{G}^*(V_s,(Z_{1,s},Z_{2,s})^{tr},(q_{1,s},q_{2,s})^{tr})ds\right\vert\mathcal{F}_t\right]\notag\\
&=&-\essinf_{q\in\mathcal{A}^*_{[t,T]},
q_{1,\cdot}=-\theta(V_{\cdot}) }E_{\mathbb{Q}^q}
\left[\left.-g(V_T)+\int_t^{T}\frac{|q_{2,s}-Z_{2,s}|^2}{2\gamma}ds\right\vert\mathcal{F}_t\right].
\end{eqnarray}

Note that in this case, the first component $q_{1}$ of the density
process $q$ is the negative market price of risk $-\theta(V_s)$, for
$s\in[0,T]$. Hence, under $\mathbb{Q}^q$, the stock price process
$S$ follows
\begin{equation*}
dS_t=\sigma(V_t)S_tdW_t^{{q},1},
\end{equation*}
where $W^{{q},1}=W^{1}+\int_0^{\cdot}\theta(V_s)ds$ is a Brownian
motion under $\mathbb{Q}^q$. Therefore, $\mathbb{Q}^q$ is an
equivalent martingale measure.

In this case, the penalty term $G^*$ reduces to a quadratic function
of the density process $q_{2}$ and the forward volatility process
$Z_2$.

\subsection{Classical entropic risk measures}

For the classical entropic risk measure, we have the
representation $\rho _{t,T}(\xi _{T})=P_{t}^{-\xi _{T}}-P_{t}^{0},$ where $%
P_{t}^{-\xi _{T}}$, $t\in[0,T]$, is the unique solution to
(\ref{BSDE_11}) with $\bar{\xi}_T$ replaced by $\xi_T$,
\begin{eqnarray}
&&P_{t}^{-{\xi}_{T}}=-{\xi}_{T}+\int_{t}^{T}\frac{1}{\gamma }F\left(
V_{s}^v,\gamma Q_{s}^{-{\xi}_{T}}\right) ds-\int_{t}^{T}(
Q_{s}^{-{\xi}_{T}})^{tr}dW_{s}\label{BSDE_1111}\\
&&=g(V_T)+\int_t^T\frac{\gamma^2|Q_{2,s}^{-\xi_T}|^2-|\theta(V_s)|^2-2\theta(V_s)Q_{1,s}^{-\xi_T}}{2\gamma}
ds- \int_t^{T}(Q_s^{-\xi_T})^{tr}dW_s\text{.}\notag
\end{eqnarray}%

Direct calculations then yield the closed-form representation
\begin{equation}
\rho _{t,T}(\xi _{T})=P_{t}^{-\xi _{T}}-P_{t}^{0}=\frac{1}{\gamma
|\kappa _{2}|^{2}}\ln E_{\mathbb{Q}^{T}}\left[ e^{\gamma |\kappa
_{2}|^{2}g\left( V_{T}\right) }|\mathcal{F}_{t}\right] \text{,}
\label{closed-form-backward}
\end{equation}%
where the measure $\mathbb{Q}^{T}$, defined on $\mathcal{F}_{T},$ is
equivalent to $\mathbb{P}$ and satisfies
\begin{equation}
\left.\frac{d\mathbb{Q}^{T}}{d\mathbb{P}}\right|_{\mathcal{F}_{t}}:=\mathcal{E%
}\left( -\int_{0}^{\cdot }\left( \kappa _{1}\theta \left(
V_{s}\right) -\kappa _{2}\gamma Q_{2,s}^{0}\right) (\kappa
_{1}dW_{s}^{1}+\kappa _{2}dW_{s}^{2})\right) _{t}\text{.}
\label{Q-backward}
\end{equation}

Note that in this single stock/single factor example, the only
difference between the forward and classical entropic risk measures
is their respective
measures $\mathbb{Q}$ and $\mathbb{Q}^{T}$ (cf. (\ref{Q-forward}) and (\ref%
{Q-backward})).

In the forward case, the pricing measure $\mathbb{Q}$ is determined
by the component $Z_{2},$ appearing in the ergodic BSDE
representation of the forward performance process
(\ref{ExponentialForwardUtility}). It is naturally
\textit{independent} of the maturity $T$. In the classical case,
however, the pricing measure $\mathbb{Q}^{T}$ is determined by the
component
$Q_{2}^{0}$ , coming from the exponential utility maximization (\ref{def_2}%
) with zero risk position (cf. (\ref{BSDE_1111}) with $\xi_T=0$),
which depends critically on the maturity $T$.

For $t\in \left[ 0,T\right] ,$ the two measures are
related as%
\begin{equation*}
\left. \frac{d\mathbb{Q}^{T}}{d\mathbb{Q}}\right \vert _{\mathcal{F}%
_{t}}=\left. \frac{d\mathbb{Q}^{T}}{d\mathbb{P}}\right \vert _{\mathcal{F}%
_{t}}\left( \left. \frac{d\mathbb{Q}}{d\mathbb{P}}\right \vert _{\mathcal{F}%
_{t}}\right) ^{-1}
\end{equation*}%
\begin{equation*}
=e^{\int_{0}^{t}\kappa _{1}\theta (V_{s})\kappa _{2}(\gamma
Q_{2,s}^{0}-Z_{2,s})ds}\mathcal{E}\left( \int_{0}^{\cdot }\kappa
_{2}\left( \gamma Q_{2,s}^{0}-Z_{2,s}\right) (\kappa
_{1}dW_{s}^{1}+\kappa _{2}dW_{s}^{2})\right) _{t}.
\end{equation*}

We conclude with numerical results for $\rho _{0}(\xi _{T})$ and $%
\rho _{0,T}(\xi _{T}),$ taking $T$ as large as possible, with $\eta
(v)=-\alpha v$, $\theta (v)=(K_{2}-|v|)_{+}$, and
$g(v)=(K_{1}-|v|)_{+}$, for two positive constants $K_{1},K_{2}$. The
graphs in Figures 1, 2, and 3, with different values of $ \kappa $, confirm the long maturity behavior of both
the forward and classical entropic risk measures. However, it is not
clear what the relationship between such two limiting constants is. Moreover, the graphs confirm that the limiting constants are indeed independent of the initial value of the stochastic factor process.

\appendix
\section{Appendix}

\subsection{Estimates on the stochastic factor process $V$}

\begin{proposition}
\label{prop:forward_estimate} Suppose that Assumption
\ref{assumption2} holds and that $V_0^v=v$. Then, the following
assertions hold, for all $t\geq 0.$

i) The stochastic factor process satisfies $|V_{t}^{v}-V_{t}^{\bar{v}%
}|^{2}\leq e^{-2C_{\eta }t}|v-\bar{v}|^{2}$, for $v,\bar{v}\in
\mathbb{R}^{d}.$

ii) Assume that the process $V^{v}$ follows
\begin{equation*}
dV_{t}^{v}=\left( \eta (V_{t}^{v})+H(V_{t}^{v})\right) dt+\kappa dW_{t}^{H}%
\text{,}
\end{equation*}%
where $H:\mathbb{R}^{d}\rightarrow \mathbb{R}^{d}$ is a measurable
bounded function, $\mathbb{Q}^{H}$ is a probability measure
equivalent to $\mathbb{P} $, and $W^{H}$ is a
$\mathbb{Q}^{H}$-Brownian motion. Then, there exists a constant
$C>0$ such that $E_{^{\mathbb{Q}^{H}}}\left[ |V_{t}^{v}|^{p}\right]
\leq
C\left( 1+|v|^{p}\right) $, for any $p\geq 1$. Furthermore, for any measurable function $\phi :%
\mathbb{R}^{d}\rightarrow \mathbb{R}$ with polynomial growth rate
$\mu >0$, and $v,\bar{v}\in \mathbb{R}^{d},$
\begin{equation*}
\left \vert E_{^{\mathbb{Q}^{H}}}\left[ \phi (V_{t}^{v})-\phi (V_{t}^{\bar{v}%
})\right] \right \vert \leq C\left( 1+|v|^{1+\mu }+|\bar{v}|^{1+\mu
}\right) e^{-\hat{C}_{\eta }t}\text{.}
\end{equation*}%
The constants $C$ and $\hat{C}_{\eta }$ depend on the function $H$
only through $\sup_{v\in \mathbb{R}^{d}}|H(v)|$.
\end{proposition}

The proof of (i) follows from Gr{\"{o}}nwall's inequality. The first
assertion in (ii) is an application of a Lyapunov argument (see
Lemma 3.1 of \cite{FMc}), while the rest follows from the basic
coupling estimate in Lemma 3.4 in \cite{Hu11}.

\subsection{Estimates on the auxiliary function $\hat{z}(\cdot,\cdot)$}\label{appendix3}

Recall that $(Y_s,Z_s)=(y(V_s),z(V_s))$ and $
(Y_{s}^{-\xi_T},Z_{s}^{-\xi_T})=(y^{T,g}(V_s,s),z^{T,g}(V_s,s))$
{for} $s\in[0,T]. $ Hence, the pair
$(\hat{y}(V_s,s),\hat{z}(V_s,s)),$ defined as
\begin{equation}\label{auxilary_process_2}
\left(\hat{y}(V_s,s),\hat{z}(V_s,s)\right):=\left( Y_{s}^{-\xi_T}+\frac{%
Y_{s}-\lambda s}{\gamma },Z_{s}^{-\xi_T}+\frac{Z_{s}}{\gamma
}\right) ,
\end{equation}%
solves the finite horizon quadratic BSDE
\begin{equation}
{P}_{s}^{-\bar{\xi}_T}=g(V_{T})+\frac{y(V_T)-\lambda T}{\gamma }+\int_{s}^{T}\frac{%
1}{\gamma }F(V_{u},\gamma
Q^{-\bar{\xi}_T})du-\int_{s}^{T}(Q^{-\bar{\xi}_T}_u)^{tr}dW_{u}\text{,}
\label{Auxilary_BSDE}
\end{equation}%
with the driver $F(\cdot,\cdot)$ as in (\ref{driver}) and
$-\bar{\xi}_T=g(V_{T})+\frac{y(V_T)-\lambda T}{\gamma }$.

From the boundedness of $Y^{-\xi_T}$ and the linear growth of $Y$,
we deduce that $\hat{y}(V_s,s)$, $s\in[0,T]$, is square-integrable,
\begin{equation*}
\sup_{s\in[0,T]}E_{\mathbb{P}}[|\hat{y}(V_s,s)|^2]\leq
C(1+\sup_{s\in[0,T]}E_{\mathbb{P}}[|V_s|^2])<\infty.
\end{equation*}
In addition, since $Z^{-\xi_T}\in\mathcal{L}^{2}_{BMO}[0,T]$ and $Z$
is bounded, we obtain that $\hat{z}(V_s,s)$, $s\in[0,T]$, is in
$\mathcal{L}^{2}_{BMO}[0,T]$.

\begin{lemma}
\label{lemma} Suppose that Assumptions \ref{assumption1} and \ref%
{assumption2} hold, and that the risk position $\xi _{T}$ is as in
(\ref{payoff-european}). Then, the following assertions hold:

i) There exists a unique solution
$(P^{-\bar{\xi}_T},Q^{-\bar{\xi}_T})$ to equation
(\ref{Auxilary_BSDE}), where $P^{-\bar{\xi}_T}$ is
square-integrable, i.e.
$\sup_{s\in[0,T]}E_{\mathbb{P}}[|P_s^{-\bar{\xi}_T}|^2]<\infty$, and
$Q^{-\bar{\xi}_T}\in\mathcal{L}^{2}_{BMO}[0,T]$.

ii) Moreover, the solution $Q^{-\bar{\xi}_T}$ is uniformly bounded,
namely,
\begin{equation}
|Q^{-\bar{\xi}_T}_s|\leq q\text{ \  \ with \ }q=\frac{\gamma C_{\eta
}C_{g}+C_{v}}{\gamma (C_{\eta }-C_{v})}+\frac{C_{\eta }C_{v}}{\gamma
(C_{\eta }-C_{v})^{2}}\text{,}  \label{q-constant}
\end{equation}%
where $C_{v}$ is given in (\ref{driver0}), and $C_{\eta }$ and
$C_{g}$ in Assumptions \ref{assumption2} and
(\ref{payoff-european}), respectively. Hence, the process
$\hat{z}(V_s,s)$, $s\in[0,T]$, given in (\ref{auxilary_process_2})
is also uniformly bounded by $q$.
\end{lemma}

\begin{proof}
(i) The existence of the solutions follows from the construction
(\ref{auxilary_process_2}). We next establish uniqueness.

To this end, let $(P^{-\bar{\xi}_T},Q^{-\bar{\xi}_T})$ and
$(\bar{P}^{-\bar{\xi}_T},\bar{Q}^{-\bar{\xi}_T})$ be two solutions
of (\ref{Auxilary_BSDE}). Let $\Delta
P^{-\bar{\xi}_T}_s:=P_s^{-\bar{\xi}_T}-\bar{P}_s^{-\bar{\xi}_T}$ and
$\Delta
Q^{-\bar{\xi}_T}_s:=Q_s^{-\bar{\xi}_T}-\bar{Q}_s^{-\bar{\xi}_T}$.
Then the pair $(\Delta P^{-\bar{\xi}_T},\Delta Q^{-\bar{\xi}_T})$
solves
\begin{eqnarray*}
\Delta P^{-\bar{\xi}_T}_s&=&\int_s^T\frac{1}{\gamma}\left(
F(V_u,\gamma
Q_u^{-\bar{\xi}_T})-F(V_u,\gamma\bar{Q}_u^{-\bar{\xi}_T})
\right)du-\int_s^T(\Delta Q_u^{-\bar{\xi}_T})^{tr}dW_u\\
&=&-\int_s^T(\Delta Q_u^{-\bar{\xi}_T})^{tr}(dW_u-\bar{M}_udu),
\end{eqnarray*}
where $\bar{M}$ is defined as
\begin{equation*}
\bar{M}_s:=\frac{\left( F(V_{s},\gamma Q_s^{-\bar{\xi}_T})-F(V_{s},\gamma \bar{Q}_s^{-\bar{\xi}_T})\right) \Delta Q_s^{-\bar{\xi}_T} }{%
\gamma \left \vert \Delta Q_s^{-\bar{\xi}_T}\right \vert
^{2}}\mathbf{1}_{\left \{\Delta Q_s^{-\bar{\xi}_T} \neq 0\right \}
}\text{.}
\end{equation*}%
Since $|\bar{M}_s|\leq
C(1+|Q_s^{-\bar{\xi}_T}|+|\bar{Q}_s^{-\bar{\xi}_T}|)$ and
$Q^{-\bar{\xi}_T},
\bar{Q}^{-\bar{\xi}_T}\in\mathcal{L}^{2}_{BMO}[0,T]$, we deduce that
$\int_0^{\cdot}(\bar{M}_s)^{tr}dW_s$ is a \textit{BMO}-martingale.
We can therefore define $%
W_{s}^{\bar{M}}:=W_{s}-\int_{0}^{s}\bar{M}_{u}du,$ which is a
Brownian motion under some probability measure
$\mathbb{Q}^{\bar{M}}$ equivalent to $\mathbb{P} $, defined as
$\frac{d\mathbb{Q}^{\bar{M}}}{d\mathbb{P}}=\mathcal{E}(\int_0^{\cdot}(\bar{M}_s)^{tr}dW_s)_T$.
Hence, we obtain that
\begin{equation*}
\Delta P_s^{-\bar{\xi}_T}=-\int_s^{T}(\Delta
Q_u^{-\bar{\xi}_T})^{tr}dW_u^{\mathbb{Q}^{\bar{M}}}.
\end{equation*}
Since $\int_0^{\cdot}(\Delta Q_u^{-\bar{\xi}_T})^{tr}dW_u$ is a
\emph{BMO}-martingale under $\mathbb{P}$ and
$\mathbb{P}\sim\mathbb{Q}^{\bar{M}}$, it follows that
$\int_0^{\cdot}(\Delta Q_u^{-\bar{\xi}_T})^{tr}dW_u^{\bar{M}}$ is a
\emph{BMO}-martingale under $\mathbb{Q}^{\bar{M}}$ (see, for
example, pp. 1563 in \cite{XYZ}), which further implies that $\Delta
P^{-\bar{\xi}_T}$ is a martingale under $\mathbb{Q}^{\bar{M}}$. The
uniqueness of the solution to (\ref{Auxilary_BSDE}) then follows by
noting that $\Delta P_T^{-\bar{\xi}_T}=0$.

(ii) For fixed $t\in[0,T]$, we consider the stochastic factor
process starting from $V_t^{t,v}=v$.
With a slight abuse of notation, we introduce the truncation function $q:%
\mathbb{R}^{d}\rightarrow \mathbb{R}^{d}$,
\begin{equation}
q(z):=\frac{\min (|z|,q)}{|z|}z\mathbf{1}_{\{z\neq 0\}}\text{,}
\label{q-quantity}
\end{equation}%
as well as the truncated version of equation (\ref{Auxilary_BSDE}),
\begin{equation}\label{truncation_equ}
P_{t}^{-\bar{\xi}_T}=g(V_{T}^{t,v})+\frac{y(V_T^{t,v})-\lambda T}{\gamma }%
+\int_{t}^{T}\frac{1}{\gamma }F(V_{s}^{t,v},\gamma q(Q
_{s}^{-\bar{\xi}_T}))ds-\int_{t}^{T}({Q}_{s}^{-\bar{\xi}_T})^{tr}dW_{s}\text{,}
\end{equation}%
whose solution is denoted as
$(\bar{y}(V_s^{t,v},s),\bar{z}(V_s^{t,v},s))$, $s\in[t,T]$.

From the form of the driver (\ref{driver}) and (\ref{q-quantity}) we
deduce the inequalities
\begin{equation}
|F(v,\gamma q(z))-F(\bar{v},\gamma q(z))|\leq C_{v}(1+\gamma q)|v-\bar{v}|%
\text{,}  \label{driver00}
\end{equation}%
\begin{equation}
|F(v,\gamma q(z))-F(v,\gamma q(\bar{z}))|\leq C_{z}(1+2\gamma q)\gamma |z-%
\bar{z}|\text{,}  \label{driver11}
\end{equation}%
for any $v,\bar{v},z,\bar{z}\in \mathbb{R}^{d}$. Next, we consider
the above truncated equation (\ref{truncation_equ}) with different
starting points $V_t^{t,v}=v$ and $V_t^{t,\bar{v}}=\bar{v}$,
\begin{eqnarray*}
&&\bar{y}(V_t^{t,v},t)-\bar{y}(V_t^{t,\bar{v}},t)\\
&=&g(V_{T}^{t,v})-g(V_{T}^{t,\bar{%
v}})+\frac{1}{\gamma }(y(V_T^{t,v})-y(V_T^{t,\bar{v}}))\\
&&+\int_{t}^{T}\frac{F(V_{s}^{t,v},\gamma q(\bar{z}(V_s^{t,v},s)
))-F(V_{s}^{t,\bar{v}},\gamma q(\bar{z}(V_s^{t,\bar{v}},s)))%
}{\gamma} ds\\
&&-\int_{t}^{T}\left( \bar{z}(V_s^{t,v},s)-\bar{z}(V_s^{t,\bar{v}},s)%
\right) ^{tr}dW_{s}.
\end{eqnarray*}
For $s\in \left[ t,T\right],$ we define the process $M_{s} $ as
\begin{equation*}
M_{s}:=\frac{\left( F(V_{s}^{t,\bar{v}},\gamma
q(\bar{z}(V_s^{t,{v}},s)))-F(V_{s}^{t,\bar{v}},\gamma
q(\bar{z}(V_s^{t,\bar{v}},s)
))\right) \left( \bar{z}(V_s^{t,{v}},s)-\bar{z}(V_s^{t,\bar{v}},s)\right)}{%
\gamma \left \vert
\bar{z}(V_s^{t,{v}},s)-\bar{z}(V_s^{t,\bar{v}},s)\right \vert ^{2}},
\end{equation*}%
whenever $\bar{z}(V_s^{t,{v}},s)-\bar{z}(V_s^{t,\bar{v}},s)\neq 0$,
and $0$ otherwise. In turn,
\begin{eqnarray*}
&&\bar{y}(V_t^{t,v},t)-\bar{y}(V_t^{t,\bar{v}},t)\\
&=&g(V_{T}^{t,v})-g(V_{T}^{t,\bar{v}})+\frac{1}{\gamma }%
(y(V_T^{t,v})-y(V_{T}^{t,\bar{v}}))\\
&&+\int_{t}^{T}\frac{F(V_{s}^{t,v},\gamma q(\bar{z}(V_s^{t,v},s)
))-F(V_{s}^{t,\bar{v}},\gamma q(\bar{z}(V_s^{t,{v}},s)))%
}{\gamma} ds\\
&&-\int_{t}^{T}\left( \bar{z}(V_s^{t,v},s)-\bar{z}(V_s^{t,\bar{v}},s)%
\right)^{tr}(dW_{s}-M_{s}ds).
\end{eqnarray*}%

Note, however, that $M_{s}$ is bounded as it
follows from inequality\textbf{\ }(\ref{driver11}). Thus, we can define $%
W_{s}^{M}:=W_{s}-\int_{0}^{s}M_{u}du,$ which is a Brownian motion
under some measure $\mathbb{Q}^{M}$ equivalent to $\mathbb{P} $,
defined on $\mathcal{F}_{T}.$ In turn,
\begin{eqnarray*}
&& |\bar{y}(v,t)-\bar{y}(\bar{v},t)|\\
&\leq& C_{g}E_{\mathbb{Q}^{M}}\left[ |V_{T}^{t,v}-V_{T}^{t,\bar{v}}||%
\mathcal{F}_{t}\right] +\frac{C_{v}}{\gamma (C_{\eta }-C_{v})}E_{\mathbb{Q}%
^{M}}\left[ |V_{T}^{t,v}-V_{T}^{t,\bar{v}}||\mathcal{F}_{t}\right]\\
&&+\frac{C_{v}(1+\gamma q)}{\gamma }E_{\mathbb{Q}^{M}}\left[ \left.
\int_{t}^{T}|V_{s}^{t,v}-V_{s}^{t,\bar{v}}|ds\right \vert \mathcal{F}_{t}%
\right]\\
&\leq& \left( C_{g}+\frac{C_{v}}{\gamma (C_{\eta }-C_{v})}+\frac{%
C_{v}(1+\gamma q)}{\gamma C_{\eta }}\right) |v-\bar{v}|\text{,}
\end{eqnarray*}%
where we used the Lipschitz continuity conditions on $g(v)$, $y(v)$ and $%
F(v,\gamma q(z))$ with respect to $v$ (cf. (\ref{payoff-european}), (\ref%
{gradient_for_y}) and (\ref{driver00}) respectively), and the
exponential ergodicity condition (i) in Proposition
\ref{prop:forward_estimate}.

From the relationship $\kappa ^{tr}\nabla
\bar{y}(V_{s}^{t,v},s)=\bar{z}(V_s^{t,v},s)$, we further deduce that
$q(\bar{z}(V_s^{t,v},s))=\bar{z}(V_s^{t,v},s)$, and that
$|\bar{z}(V_s^{t,v},s)|\leq q$. In other words, the truncation does
not play a role, and the pair
$(\bar{y}(V_{s}^{t,v},s),\bar{z}(V_{s}^{t,v},s))$, $s\in[t,T]$, also
solves (\ref{Auxilary_BSDE}). Therefore, $q$ is the uniform bound of
$\hat{z}(V_{s}^{t,v},s)$.
\end{proof}


\newpage

\begin{figure}\label{fig:1}
\begin{center}
\begin{subfigure}[b]{0.95\textwidth}
\begin{center}
\includegraphics[width=\textwidth]{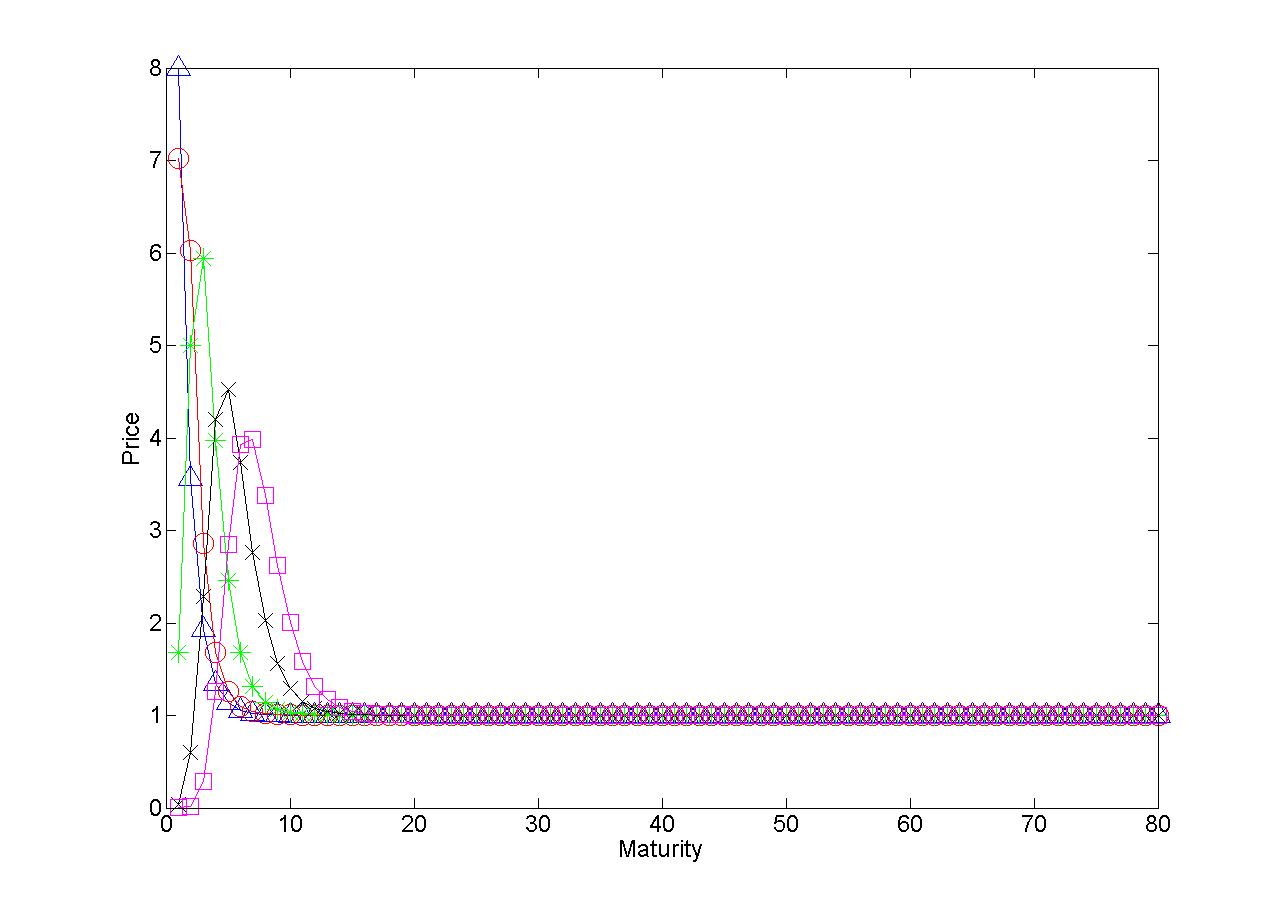}
\caption[Forward]%
{{\small Forward entropic risk measure}}
\end{center}
\end{subfigure}
\quad
\begin{subfigure}[b]{0.95\textwidth}
\begin{center}
\includegraphics[width=\textwidth]{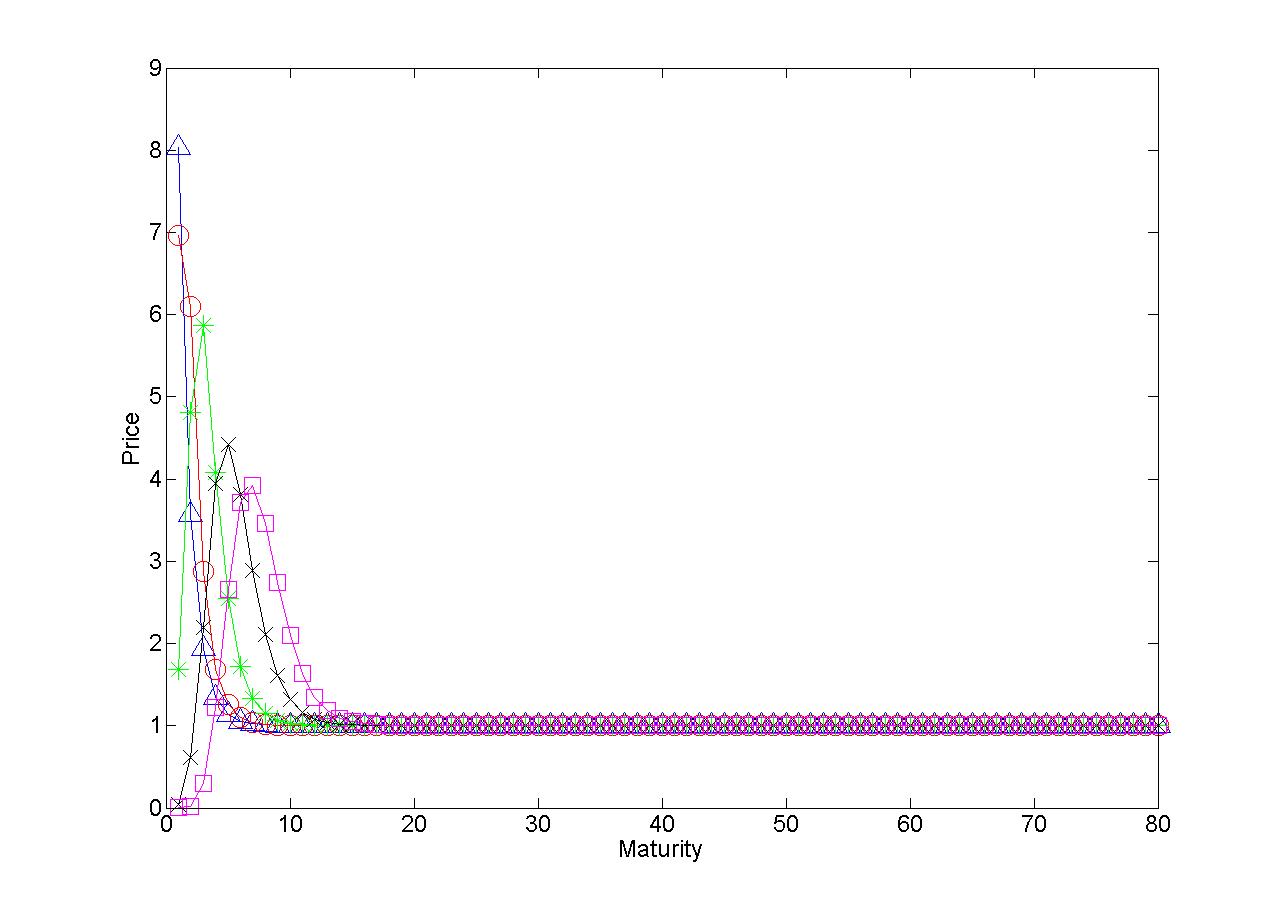}
\caption[Classical]%
{{\small Classical entropic risk measure}}
\end{center}
\end{subfigure}
\caption[$ T $-normalized forward entropic risk measure against the maturity $ T $]
{\small Forward and classical entropic risk measures against the maturity $ T $, with $ \gamma=1 $, $ \alpha=0.1 $, $ K_1=K_2=10 $, $ \kappa_1=0.9 $, $ \kappa_2=0.1 $; blue upward-pointing triangle for $ v_0=5 $, red circle for $ v_0=7.5 $, green asterisk for $ v_0=10 $, black cross for $ v_0=12.5 $ and magenta square $ v_0=15 $.}
\end{center}
\end{figure}

\newpage

\begin{figure}\label{fig:2}
\begin{center}
\begin{subfigure}[b]{0.95\textwidth}
\begin{center}
\includegraphics[width=\textwidth]{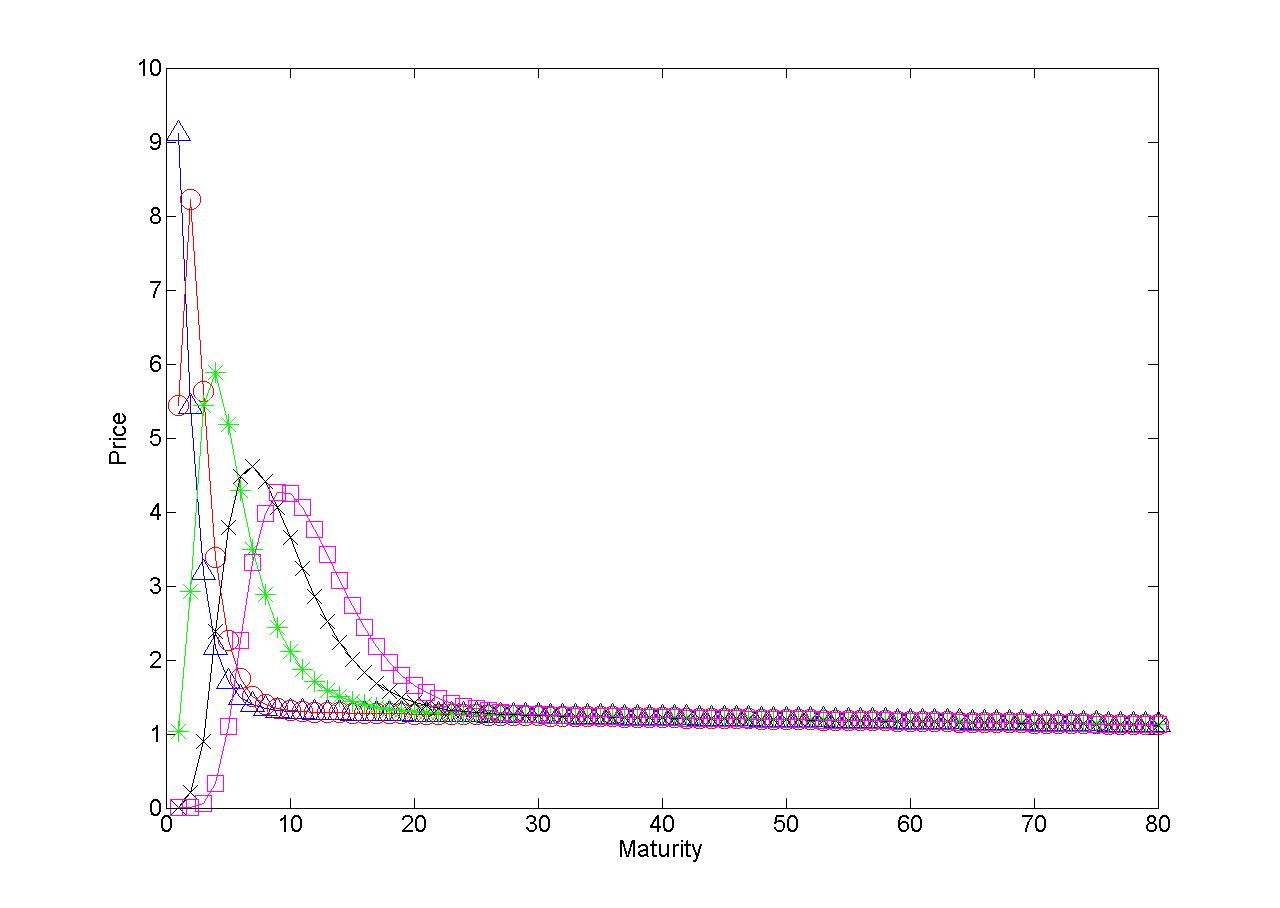}
\caption[Forward]%
{{\small Forward entropic risk measure}}
\end{center}
\end{subfigure}
\quad
\begin{subfigure}[b]{0.95\textwidth}
\begin{center}
\includegraphics[width=\textwidth]{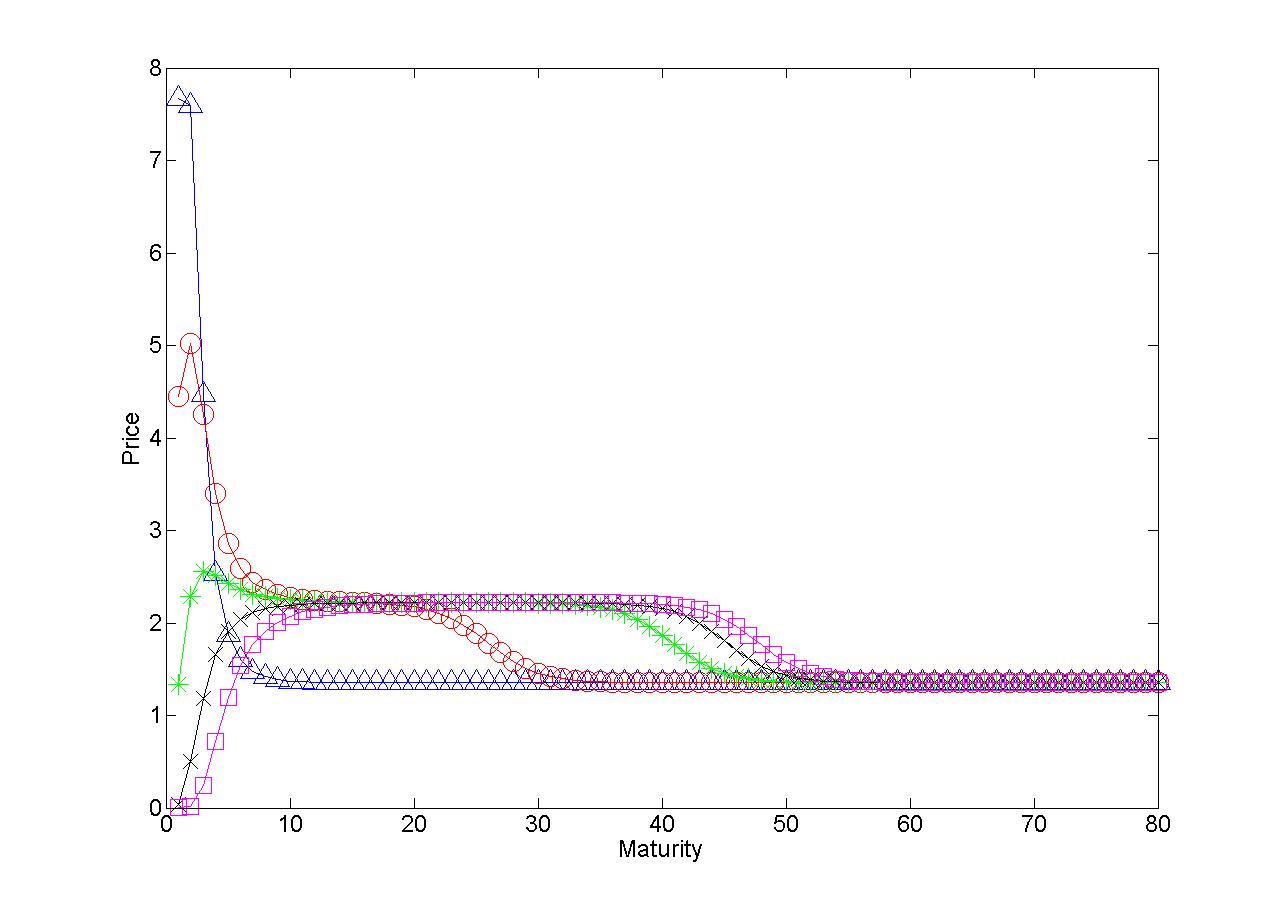}
\caption[Classical]%
{{\small Classical entropic risk measure}}
\end{center}
\end{subfigure}
\caption[$ T $-normalized forward entropic risk measure against the maturity $ T $]
{\small Forward and classical entropic risk measures against the maturity $ T $, with $ \gamma=1 $, $ \alpha=0.1 $, $ K_1=K_2=10 $, $ \kappa_1=0.5 $, $ \kappa_2=0.5 $; blue upward-pointing triangle for $ v_0=5 $, red circle for $ v_0=7.5 $, green asterisk for $ v_0=10 $, black cross for $ v_0=12.5 $ and magenta square $ v_0=15 $.}
\end{center}
\end{figure}
\newpage

\begin{figure}\label{fig:3}
\begin{center}
\begin{subfigure}[b]{0.95\textwidth}
\begin{center}
\includegraphics[width=\textwidth]{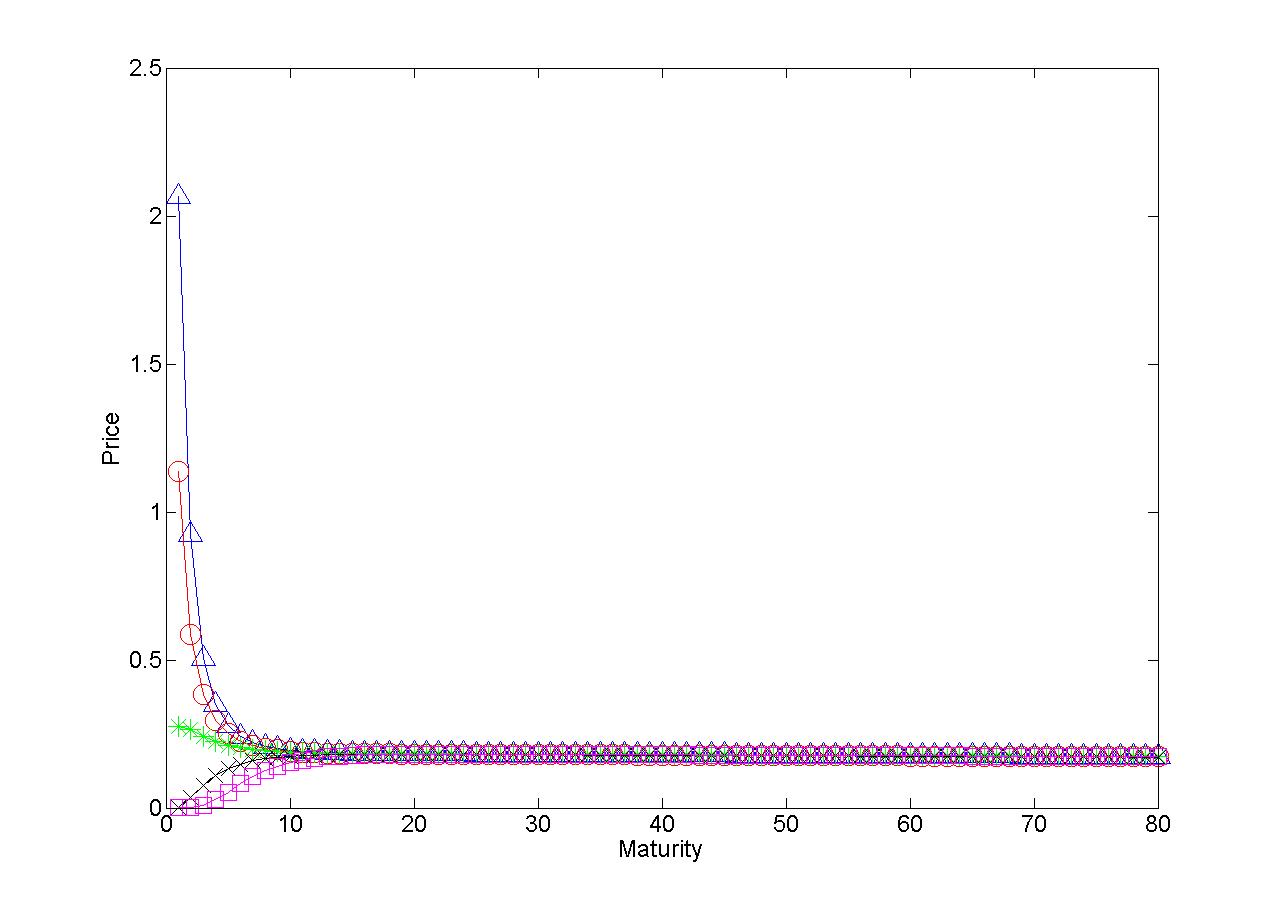}
\caption[Forward]%
{{\small Forward entropic risk measure}}
\end{center}
\end{subfigure}
\quad
\begin{subfigure}[b]{0.95\textwidth}
\begin{center}
\includegraphics[width=\textwidth]{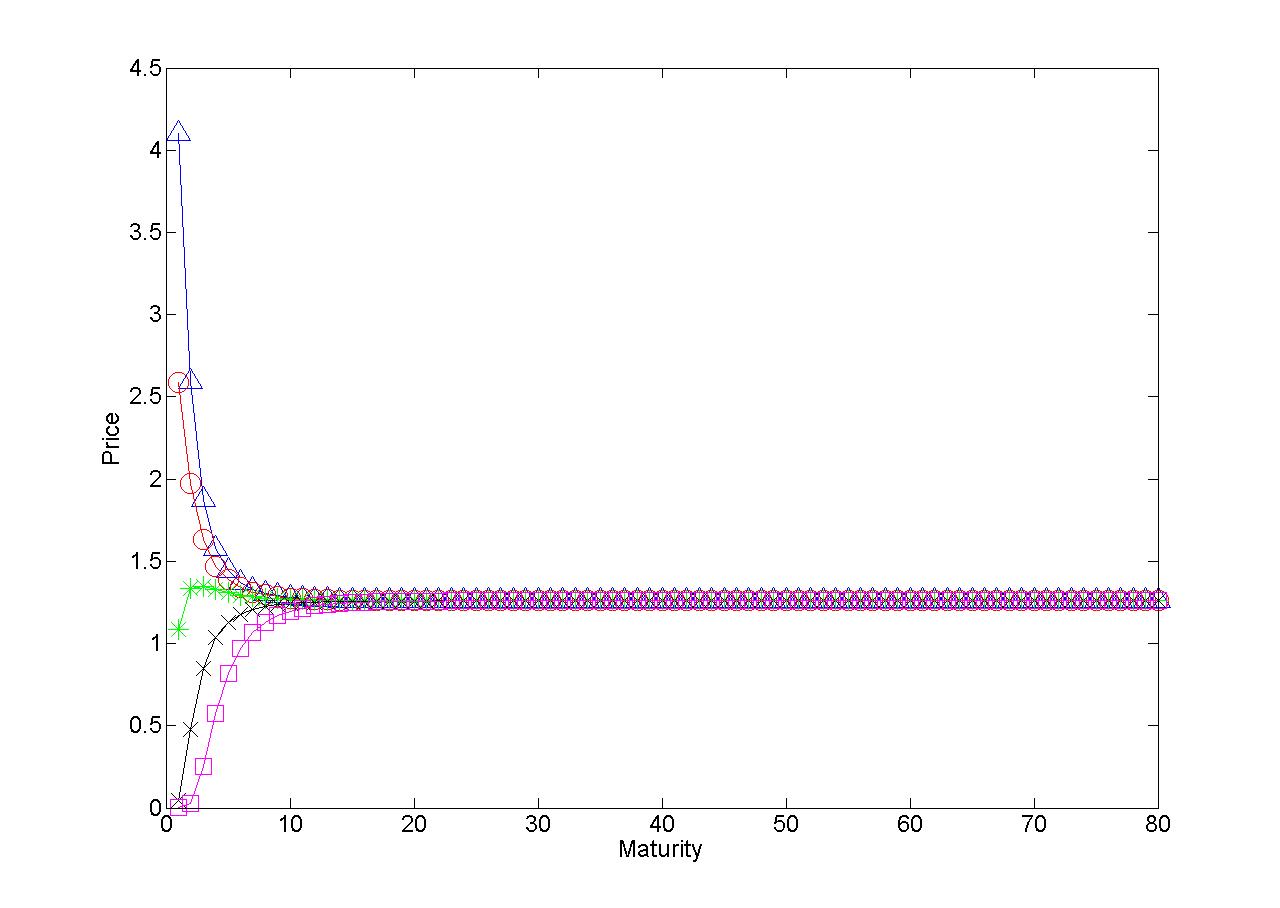}
\caption[Classical]%
{{\small Classical entropic risk measure}}
\end{center}
\end{subfigure}
\caption[$ T $-normalized forward entropic risk measure against the maturity $ T $]
{\small Forward and classical entropic risk measures against the maturity $ T $, with $ \gamma=1 $, $ \alpha=0.1 $, $ K_1=K_2=10 $, $ \kappa_1=0.0 $, $ \kappa_2=1.0 $; blue upward-pointing triangle for $ v_0=5 $, red circle for $ v_0=7.5 $, green asterisk for $ v_0=10 $, black cross for $ v_0=12.5 $ and magenta square $ v_0=15 $.}
\end{center}
\end{figure}

\end{document}